\documentclass[12pt,english]{article}
\usepackage[T1]{fontenc}
\usepackage[latin9]{inputenc}
\usepackage{geometry}
\usepackage{amsmath}
\usepackage{amsthm}

\makeatletter
\theoremstyle{plain}
\newtheorem{prop}{\protect\propositionname}
\theoremstyle{plain}
\newtheorem{lem}{\protect\lemmaname}
\theoremstyle{plain}
\newtheorem{assumption}{\protect\assumptionname}
\theoremstyle{plain}
\newtheorem{cor}{\protect\corollaryname}

\usepackage{babel}
\usepackage{amssymb}

\providecommand{\dom}{\mathrm{dom}}
\providecommand{\imp}{\mathrm{imp}}
\providecommand{\exp}{\mathrm{exp}}

\makeatother

\usepackage{babel}
\providecommand{\assumptionname}{Assumption}
\providecommand{\corollaryname}{Corollary}
\providecommand{\lemmaname}{Lemma}
\providecommand{\propositionname}{Proposition}

\begin{document}
\title{Tradeable Import Certificates: A Promising Instrument to Support Domestic
Production in Strategic Sectors?}
\author{Sebastian Kranz, Ulm University}
\date{2025-11-27}
\maketitle
\begin{abstract}
Recent crises have increased concerns about supply security in sectors
that are considered strategically important. The goal of sufficient
domestic production capacities has motivated various forms of subsidies,
tariffs and other instruments. This paper revisits Warren Buffett\textquoteright s
(2003) proposal of tradeable import certificates (TIC) in this context.
TIC differ from classical import quotas mainly by linking the import
volume to export performance. The certificate price functions like
a mix of flexible tariffs and export subsidies whose levels depend
on net imports in the strategic sector. We analyse benefits and drawbacks
in a simple two-country model. In competitive markets, TIC constitute
a transparent and efficient instrument that effectively reduces incentives
for other countries to deviate from agreements via hidden subsidies
or non-tariff trade barriers. However, TIC can have adverse effects
if there are domestic producers with market power in the certificate
market.
\end{abstract}

\section{Introduction}

Over the past decade, governments in advanced and emerging economies
alike have rediscovered a vocabulary of ``strategic sectors,'' ``resilience,''
and ``economic security.'' The Trump administration's call to bring
manufacturing back to the United States, Europe's debates about ``strategic
autonomy,'' and the scramble for vaccines, active pharmaceutical
ingredients, and semiconductors during the Covid-19 pandemic are emblematic.
Supply-chain breakdowns, fears of overreliance on geopolitical rivals,
diminished bargaining power in international negotiations, and the
need to ramp up production quickly in the event of armed conflict
or other emergencies have elevated production volumes in certain sectors
from a matter of efficiency to a matter of statecraft.\footnote{For concrete examples see, inter alia, White House (2017), European
Commission (2021), and the extensive reporting on shortages in medical
goods and chips during 2020--22.}

Governments have reacted to these strategic-sector anxieties with
an eclectic mix of subsidies and tariffs. On the spending side, the
U.S. CHIPS and Science Act allocates \$52.7 billion in grants and
tax incentives for domestic semiconductor production and related activities
(NIST, 2023); the EU's Green Deal Industrial Plan channels comparable
support to ``net-zero'' technologies (European Commission, 2023);
and China's Made in China 2025 programme deploys tax breaks, subsidised
credit and direct transfers to nurture national champions (State Council,
2015; McBride and Chatzky, 2019). By contrast, the stick of choice
has been tariffs: the first Trump administration raised average U.S.
applied tariffs to their highest level in the post-war era (Amiti,
Redding, and Weinstein, 2019), and recent tariff hikes and proposals
would push effective rates far above the pre-2017 norm (CBO, 2025).

Both instruments, however, come with well-documented side-effects.
Sharp tariff swings amplify trade-policy uncertainty, deterring investment
and export entry - as shown by Handley and Limão\textquoteright s
(2017) evidence for China-U.S. trade - while blanket subsidy programmes
invite fiscal waste and rent-seeking, a classic industrial-policy
pitfall highlighted by Juhász et al. (2024) and already visible in
the global proliferation of export credit support (Dawar, 2020). 

Moreover, such policies erode trust: partners struggle to distinguish
legitimate security goals from beggar-thy-neighbour tactics, fuelling
anti-subsidy investigations and the spectre of retaliatory tariff
hikes. If strategic objectives are here to stay, more transparent
and rules-based approaches could be beneficial. 

This paper revisits Warren Buffett\textquoteright s (2003) proposal
of tradeable import certificates (TIC), an idea that, despite its
intuitive appeal, has received surprisingly little attention in the
academic literature.\footnote{We are only aware of Papadimitriou et al. (2008) who perform a simulation
study to estimate the impact of Buffett's original proposal on US
import prices and trade deficit. They also note possible issues: price
volatility for certificates, risk of partner retaliation, and implementation
complexities, but don't provide a deeper theoretical analysis.} We examine in a theoretical model how TIC perform if the mechanism
is adapted to sectors in which countries target a level of minimum
domestic production for strategic reasons.

Early trade models, such as those by Bhagwati and Srinivasan (1969)
and Shibata (1968), show that tariffs and quotas can yield equivalent
production outcomes, although tariffs are typically preferred for
generating revenue in the importing country. Later work by Tower (1975)
and Melvin (1986) explores strategic interactions between countries
and shows how import or export quotas can trigger inefficient trade
wars, pushing countries toward autarky.

Buffett\textquoteright s TIC proposal differs from the traditional
import quotas studied in that literature in two fundamental ways.
First, the number of import certificates is not fixed exogenously
but is instead endogenously tied to export volumes---every export
unit generates $\eta$ certificates that entitle imports. Second,
these certificates are allocated to exporters, who can then sell them
to importers in a secondary market. Consequently, certificate prices
function like a combination of tariffs and export subsidies whose
level flexibly and automatically adjusts with the relative level of
import pressure vs export performance.

As a variation of Buffett's proposal, Papadimitriou et al. (2008)
suggest that the government should auction off the certificates and
use the income to reduce payroll taxes. We study a generalization
by allowing an exogenous split of certificate revenues between private
exporters and the government. Specifically, a fraction $\phi$ between
zero and one of the proceeds from certificate sales can be retained
by exporters, with the remainder accruing to the state. 

We show in our model with perfect competition and Ricardian gains
from trade that a constrained-efficient TIC mechanism grants exporters
the revenues of exactly one input certificate for every exported unit,
which is the case if $\phi\eta=1$. Then the effective rates of tariffs
and export subsidies that incorporate certificate prices are equal
in size.\footnote{For $\eta>1$ that solution lies between Buffet's proposal to grant
all certificates to exporters and given all income to the government.
For example, if one unit of exports would generated $\mu=1.5$ import
certificates the exporter should not receive all revenues from certificate
sales but only keep $2/3$ of the revenues.} This results in an industry structure where, taking into account
the strategic preferences, production of different products within
the sector is still efficiently distributed across countries. Compared
to a solution that more strongly relies on effective tariffs than
export subsidies, international trade volume remains higher.

An implicit core assumption of our analysis is that while there are
relative comparative advantages for individual goods, a severe crisis
that disrupts trade would allow sector-wide capacity to be retooled
to produce other critical products within the same sector that were
previously imported. An often mentioned example is how during World
War II, U.S. civilian industries were converted to military production,
see e.g. Rhode et al. (2018). Bernard, Redding, and Schott (2010)
show that over half of U.S. manufacturing firms switch their product
mix within a five-year period as market conditions change.

Beyond the Ricardian benefits of trade in our model, we see independent
reasons to prefer efficient intra-sector trade over equilibria with
depressed trade volumes.

First, deeper bilateral trade links can raise the opportunity cost
of conflict and, in turn, promote peaceful cooperation. Of course,
the link between trade and peaceful cooperation is not always as clear
cut: the relational contracting framework in Goldlücke \& Kranz (2023)
stresses also the importance of symmetric bargaining positions should
trade break down. A TIC scheme that guarantees that each country produces
a certain minimum share of products from the strategic sector, may
facilitate such symmetric positions with high mutual incentives for
further cooperation.\footnote{Martin, Mayer, \& Thoenig, (2008) provides a deeper exploration of
the links between trade and conflict.}

Second, protection delivered primarily through tariffs tends to dull
competitive pressure and invite rent-seeking as firms fight to secure
sheltered domestic margins rather than productivity improvements (Krueger,
1974). Empirically, trade liberalization episodes show the flip side
of this coin: when tariff shields recede, low-productivity plants
contract or exit and industry-level productivity rises (Pavcnik, 2002;
Trefler, 2004). Effective industrial policy therefore seems to work
best when it embeds \textquotedblleft discipline\textquotedblright ---linking
support to performance in contestable markets---rather than insulating
firms from rivalry. Case-based accounts emphasize export discipline
in East Asia (Studwell, 2013), and cross-country micro-evidence shows
that competition-friendly targeting is associated with higher productivity
growth (Aghion, Cai, Dewatripont, Du, Harrison, \& Legros, 2015).

One important aspect of international trade agreements is how they
affect countries incentives to game the system in particular using
hard to monitor instruments like implicit subsidies or non-tariff
trade barriers. For example, Hillman and Manak (2023) summarize: ``The
{[}WTO{]} rules did little to prevent widespread industrial subsidies
use by countries hoping to gain an edge in international trade.''
Section 3 analyzes in detail how the automatic adjustment inherent
in a TIC scheme can effectively reduce incentives of a partner country
for such deviations and limit the negative impact should such deviations
take place nevertheless.

We also discuss in Section 3 how TIC schemes may facilitate the negotiation
of international trade agreements in the presence of strategic concerns.
Unlike agreements that fix tariffs or subsidy levels to account for
strategic concerns, TIC schemes allow to implement agreements in a
symmetric fashion without relying on imperfectly observable information
about relative comparative advantages. 

Dynamic adjustment to changes in comparative advantage is an inherent
feature of TIC schemes, whereas subsidies or tariffs would need to
be repeatedly adjusted as economic fundamentals evolve. Yet, tariffs
and subsidies are often politically difficult to unwind once in place,
even when their original justification no longer holds, see e.g. Freund
and Özden (2008) or Magee (2002). With a TIC Agreement effective tariffs
and export subsidies automatically vanish once the trade deficit in
the strategic sector becomes sufficiently small.

While the points above highlight the potential benefits of TIC schemes,
our analysis in Section 3.3 highlights a crucial drawback. A TIC scheme
can have severe adverse effects if the strategic sector is small or
concentrated enough such that domestic producers can exert market
power in the certificate market. By strategically reducing exports,
producers could then raise certificate prices and benefit from the
added domestic protection and higher prices of the sold certificates.
This is an unfortunate aspect of TIC schemes, which limits the potential
applicability to sufficiently widely defined strategic sectors. For
supply security in critical but small sectors, like certain rare earth
minerals, TIC don't look like a promising tool.

The remaining paper is structured as follows. Section 2 presents the
core model: a two-country, partial-equilibrium Ricardian setting with
a continuum of perfectly competitive product markets ordered by relative
efficiency similar to Dornbusch et al. (1977). The section also contains
useful results that characterize the market outcome for given tariffs,
subsidies and certificate prices.

Section 3 performs the main analysis focusing on a case that a country
with a competitive disadvantage has the strategic goal to increase
its production in the strategic sector up to a specific target level.
We derive the main insights already summarized above. Section 4 provides
concluding remarks. An appendix contains all proofs and additional
theoretical results.

\section{Market Model and Outcomes}

We consider two countries $A$ and $B$. Throughout the paper, we
use $i,j,k\in\{A,B\}$ as country indices such that always $j\ne i$
while $k$ can be equal or unequal $i$. Every formula for country
$i$ holds symmetrically for country $j$ by swapping all $i$ and
$j$. We study a strategic sector consisting of a continuum of independent
products with unit mass indexed by $m\in[0,1]$. Each product is traded
on its own perfectly competitive market that is independent from all
other markets. We assume that consumers have a valuation $v$ for
each product and that for prices below $v$ demand is completely inelastic
and given by 1 for each product in each country. The value of $v$
is assumed be sufficiently large such that it exceeds all equilibrium
prices in the competitive goods markets.

The marginal cost for domestically consumed production by country
$i$ for product $m$ is given by
\begin{equation}
c_{i}^{dom}(m)=w_{i}(m)-s_{i}\label{eq:c_dom}
\end{equation}
where $w_{i}(m)$ is the technological marginal cost for product $m$
in country $i$. The variable $s_{i}\geq0$ is a production subsidy
paid to all units produced in country $i$. The marginal cost of an
exported product $m$ from country $i$ to country $j$ is given by
\begin{equation}
c_{i}^{exp}(m)=w_{i}(m)-s_{i}-e_{i}+\tau_{j}+\pi_{j}-\phi_{i}\eta_{i}\pi_{i}\label{eq:c_exp}
\end{equation}
where $e_{i}\geq0$ is an export subsidy and $\tau_{j}\geq0$ is the
tariff rate of country $j$, imposed uniformly on all imported units
in the strategic sector.

The variable $\pi_{i}$ denotes the equilibrium price for tradeable
import certificates (TIC) in country $i$. Every imported unit in
country $i$ will require exactly one import certificate. These import
certificates are generated by country $i$'s exports: for every exported
unit, the exporter gets $\eta_{i}$ certificates that are sold in
a perfectly competitive market to importers. We call $\eta_{i}$ the
export credit factor. While in Buffett's original proposal all proceeds
from the sold certificates go to the exporter, Papadimitriou et al.
(2008) suggest that the government could auction off all certificates
and keep the income. We allow for both and intermediate cases by assuming
that exporters can only keep a share $\phi_{i}\in[0,1]$ from the
revenues of sold import certificates; the remaining share goes to
the state. Thus the term $\phi_{i}\eta_{i}\pi_{i}$ in $c_{i}^{exp}(m)$
describes the exporter's revenues from the certificates generated
by one exported unit. If country $i$ does not implement a TIC scheme,
we fix $\pi_{i}=0$.

It is helpful to define effective tariff rates $\tilde{\tau}_{i}$
as the sum of tariff and certificate prices:
\begin{equation}
\tilde{\tau}_{i}=\tau_{i}+\pi_{i}.\label{eq:tau_tilde}
\end{equation}
Similarly, we define effective export subsidies as 
\begin{equation}
\tilde{e}_{i}=e_{i}+\phi_{i}\eta_{i}\pi_{i}.\label{eq:e_tilde}
\end{equation}
So marginal costs of exports from country $i$ to $j$ can be written
as
\begin{equation}
c_{i}^{exp}(m)=w_{i}(m)-s_{i}-\tilde{e}_{i}+\tilde{\tau}_{j}\label{eq:c_exp_tilde}
\end{equation}
The source for gains of trade in our model is given by differentiated
technological marginal costs $w_{i}(m)$ between both countries. Let
\[
\Delta w_{i}(m)\equiv w_{i}(m)-w_{j}(m)
\]
denote country $i$'s technological cost disadvantage for product
$m$ compared to country $j$. For tractability, we impose a simple
linear structure and assume that 
\begin{equation}
\Delta w_{A}(m)=-\alpha_{A}+(\alpha_{A}+\alpha_{B})\,m\label{def_Delta_w}
\end{equation}

with $\alpha_{A},\alpha_{B}>0$. This means for $m=0$ country A has
the largest competitive advantage with $\Delta w_{A}(0)=-\alpha_{A}$
while for $m=1$ country $B$ has the largest advantage with $\Delta w_{A}(1)=\alpha_{B}$
and inbetween the cost advantage of country B increases linearly in
$m$. To simplify future notation, we denote technological dispersion
by the parameter 
\begin{equation}
\delta=\alpha_{A}+\alpha_{B}.
\end{equation}

Since producers in both countries are perfectly competitive, resulting
prices for product $m$ in country $i$ are determined by the lowest
marginal costs:
\[
\min\{c_{i}^{dom}(m),c_{j}^{exp}(m)\}.
\]

Producers with the lower marginal costs serve the whole market.\footnote{In our model, the measure of product markets where domestic production
and imports have the same marginal costs is always zero.} We denote country $i$'s production of good $m$ for its domestic
market by 
\begin{equation}
q_{i}^{dom}(m)=\begin{cases}
1 & \text{if }c_{i}^{dom}(m)\leq c_{j}^{exp}(m)\\
0 & \text{otherwise}
\end{cases}\label{eq:q_dom}
\end{equation}
and the production for the export market by
\begin{equation}
q_{i}^{exp}(m)=\begin{cases}
1 & \text{if }c_{i}^{exp}(m)<c_{j}^{dom}(m)\\
0 & \text{otherwise}
\end{cases}\label{eq:q_exp}
\end{equation}

Sector wide domestic production, total exports and total imports of
country $i$ are defined as
\begin{align*}
Q_{i}^{dom} & =\int_{0}^{1}q_{i}^{dom}(m)dm,\qquad Q_{i}^{exp}=\int_{0}^{1}q_{i}^{exp}(m)dm,\qquad Q_{i}^{imp}=Q_{j}^{exp}.
\end{align*}
Note that $Q_{A}^{dom}+Q_{B}^{exp}=Q_{B}^{dom}+Q_{A}^{exp}=1$, which
implies 
\begin{equation}
Q_{A}^{dom}-Q_{A}^{exp}=Q_{B}^{dom}-Q_{B}^{exp}.\label{eq:Q_dom_Q_exp}
\end{equation}
So in our model both countries will always have the same difference
between production for domestic and export markets.

Countries $i$'s total production is denoted by
\[
X_{i}=Q_{i}^{dom}+Q_{i}^{exp}
\]
While $0\leq Q_{i}^{dom},Q_{i}^{exp}\leq1$ we have $0\leq X_{i}\leq2$.

The following result shows how in a market equilibrium $Q_{i}^{dom}$
and $Q_{i}^{exp}$ depend on the effective tariffs and subsidies of
both countries. 
\begin{prop}
\label{prop:market_eq}In every market equilibrium domestic and export
production for both countries $i\in\{A,B\}$ satisfy
\[
Q_{i}^{dom}=\max(0,\min(1,\tilde{Q}_{i}^{dom})),\qquad Q_{i}^{exp}=\max(0,\min(1,\tilde{Q}_{i}^{exp}))
\]
with
\begin{align*}
\tilde{Q}_{i}^{dom} & \equiv Q_{i}^{o}+\frac{(s_{i}-s_{j})+(\tilde{\tau}_{i}-\tilde{e}_{j})}{\delta},\\
\tilde{Q}_{i}^{exp} & \equiv Q_{i}^{o}+\frac{(s_{i}-s_{j})+(\tilde{e}_{i}-\tilde{\tau}_{j})}{\delta}
\end{align*}
where $Q_{i}^{o}$ denotes the free trade level of $Q_{i}^{dom}$
and $Q_{i}^{exp}$ and is given by
\[
Q_{i}^{o}\equiv\frac{\alpha_{i}}{\delta}.
\]
\end{prop}
Our results below sometimes refer to an \textit{interior solution},
which shall mean that for both countries $i\in\{A,B\}$ we have $Q_{i}^{dom}=\tilde{Q}_{i}^{dom}$
and $Q_{i}^{exp}=\tilde{Q}_{i}^{exp}$. For interior solutions, effective
tariffs and subsidies shift the quantities in a simple linear and
additive manner in the expected direction: production subsidies $s_{i}$
increase both domestic and export sales $Q_{i}^{dom}$ and $Q_{i}^{exp}$
while foreign production subsidies $s_{j}$ work directly opposite.
Domestic sales $Q_{i}^{dom}$ are increased by own effective tariffs
$\tilde{\tau}_{i}$ and are reduced by foreign effective exports subsidies
$\tilde{e}_{j}$, while exports $Q_{i}^{exp}$ fall in foreign effective
tariffs $\tilde{\tau}_{j}$ and increase in own effective export subsidies
$\tilde{e}_{i}$. 

\subsubsection*{TIC markets}

If country $i$ implements a tradeable import certificate (TIC) scheme,
a market equilibrium must satisfy 

\begin{equation}
\eta_{i}Q_{i}^{exp}\geq Q_{i}^{imp}\label{eq:TIC_cond}
\end{equation}
 where $\eta_{i}Q_{i}^{exp}$ is the certificate supply and $Q_{i}^{imp}$
the certificate demand. We distinguish three types of equilibria for
the TIC market of country $i$. A binding TIC equilibrium is characterized
by $\pi_{i}\geq0$ and $\eta_{i}Q_{i}^{exp}=Q_{i}^{imp}>0$; a non-binding
equilibrium by $\pi_{i}=0$, $\eta_{i}Q_{i}^{exp}\geq Q_{i}^{imp}$
and $Q_{i}^{exp}>0$. Finally, a pure autarky equilibrium is characterized
by no import or exports $Q_{i}^{exp}=Q_{i}^{imp}=0$. The whole economy
is in a market equilibrium if the TIC market for every country that
has implemented a TIC scheme satisfies one of those equilibrium conditions
and all previous equilibrium conditions for every product market $m\in[0,1]$
are satisfied. The appendix provides a series of helpful results that
characterize TIC equilibria. E.g. if both countries implement a TIC
scheme and $\eta_{A}\eta_{B}>1$ then at most one country can have
a binding TIC equilibrium, while if $\eta_{A}\eta_{B}<1$ only a pure
autarky equilibrium exists. Here is a characterization of equilibrium
certificate prices:
\begin{lem}
\label{lem:pi_binding-1} Consider a market equilibrium with an interior
solution in which country $i$ has a binding TIC equilibrium with
$\pi_{i}>0$ and the certificate price in the other country $j$ satisfies
$\pi_{j}=0$ (or $j$ has no TIC). Then the certificate price in country
$i$ is given by 
\begin{equation}
\pi_{i}=\frac{\alpha_{j}-\eta_{i}\alpha_{i}+(1+\eta_{i})(s_{j}-s_{i})+(e_{j}-\eta_{i}e_{i})+(\eta_{i}\tau_{j}-\tau_{i})}{\,1+\phi_{i}\eta_{i}^{2}\,}.\label{eq:pi_bind-1}
\end{equation}
\end{lem}
The result directly shows, how larger tariffs and subsidies by country
$j$ lead to an automatic adjustment in the certificate prices of
country $i$ if it already has a binding TIC equilibrium. Effects
become cleanest for the special case $\eta_{i}=\phi_{i}=1$. We then
find that an increase of country $j$'s tariff $\tau_{j}$ or its
export subsidies $e_{j}$ by one unit increases country $i$'s certificate
price $\pi_{i}$ and consequently its effective export subsidies $\tilde{e}_{i}$
and tariffs $\tilde{\tau}_{i}$ all by half a unit. For a unit increase
in $j$'s production subsidies $s_{j}$ the corresponding effects
are twice as large. Intuitively, that is because $s_{j}$ boosts both
domestic production and export for country $j$ and thus yields a
twice as large adjustment in country $i$'s certificate price.

\subsubsection*{Direct economic costs and countries' utility functions}

In our model each country will balance a strategic motive for higher
total production with the direct cost borne by its consumers and tax
payers minus profits of domestic firms.\footnote{Even though under perfect competition firms' profits will always be
zero, the formulas below become simpler if we directly consider firms'
profits. In particular, for imports the complete certificate price
$\pi_{i}$ enters positively, not only the share accruing to the government,
while for exports the effective subsidy rate $\tilde{e}_{i}$ becomes
relevant for direct costs, not only the direct subsidies $e_{i}$. } 

We denote by
\begin{equation}
d_{i}^{dom}(m)=c_{i}^{dom}(m)+s_{i}=w_{i}(m)\label{eq:d_dom}
\end{equation}
the direct cost for a market $m$ that is served by country $i$'s
domestic producers. Production subsidies $s_{i}$ will reduce consumers'
prices one-to-one in our model of perfect competition and thus cancel
out of these direct costs. The term

\begin{equation}
d_{i}^{imp}(m)=c_{j}^{exp}(m)-\tau_{i}-\pi_{i}=w_{j}(m)-s_{j}-\tilde{e}_{j}\label{eq:d_imp}
\end{equation}

denotes the direct cost of product $m$ if it is imported to country
$i$. Tariffs $\tau_{i}$ and certificate prices $\pi_{i}$ cancel
out as they directly increase consumer prices in country $i$. In
contrast, effective foreign subsidies $\tilde{e}_{j}$ and $s_{j}$
reduce the direct cost of imported products.

Finally,
\begin{equation}
d_{i}^{exp}(m)=\tilde{e}_{i}+s_{i}\label{eq:d_exp}
\end{equation}
denotes the direct cost for country $i$ if it exports to market $j$.
It just consists of country $i$'s effective subsidies $\tilde{e}_{i}+s_{i}$,
which benefit foreign consumers through lower prices.

We denote by
\begin{equation}
D_{i}=\int_{0}^{1}\left(d_{i}^{dom}(m)\cdot q_{i}^{dom}(m)+d_{i}^{imp}(m)\cdot q_{j}^{exp}(m)+d_{i}^{exp}(m)\cdot q_{i}^{exp}(m)\right)dm
\end{equation}
the total direct economic costs of country $i$. Let $D_{i}^{o}$
denote the direct cost of country $i$ in the free-trade outcome. 
\begin{prop}
\label{prop:E}For a given market outcome let $E_{i}$ denote the
excess cost for country $i$ compared to the free trade outcome. It
is given by
\begin{equation}
E_{i}\equiv D_{i}-D_{i}^{o}=\underbrace{\frac{\delta}{2}\big(Q_{i}^{dom}-Q_{i}^{o}\big)^{2}}_{\text{production inefficiency}}\;+\;\underbrace{(s_{i}+\tilde{e}_{i})\,Q_{i}^{exp}-(s_{j}+\tilde{e}_{j})\,Q_{i}^{imp}.}_{\text{budget transfers}}\label{eq:E_i}
\end{equation}
The total excess cost \textup{$E\equiv E_{A}+E_{B}$} of both countries
satisfies 
\begin{align}
E & =\frac{\delta}{2}\left[\big(Q_{A}^{dom}-Q_{A}^{o}\big)^{2}+\big(Q_{B}^{dom}-Q_{B}^{o}\big)^{2}\right]\label{eq:eq_E}\\
 & =\frac{\delta}{2}\left[\big(Q_{A}^{exp}-Q_{A}^{o}\big)^{2}+\big(Q_{B}^{exp}-Q_{B}^{o}\big)^{2}\right]\\
 & =\frac{\delta}{2}\left[\big(Q_{i}^{dom}-Q_{i}^{o}\big)^{2}+\big(Q_{i}^{exp}-Q_{i}^{o}\big)^{2}\right]
\end{align}
\end{prop}
The result implies that total costs are minimized in the free trade
outcome and that absent strategic concerns for higher production levels,
no country has an incentive to implement policy instruments that deviate
from the free trade outcome. 

However, if there are strategic preferences for larger domestic production
$X_{i}$, a country $i$ may use subsidies or tariffs, even if they
increase direct costs $D_{i}$. Starting from the free trade level,
a marginal increase in subsidies already has a non-zero marginal effect
on country $i$'s direct cost. That is because the subsidy payments
directly benefit foreign consumers whose import prices fall. In contrast,
the marginal effect of tariffs on a country's direct costs is initially
zero since tariffs have a 100\% pass-through rate to domestic consumers
in our model. Yet, as shown in Proposition 1, tariffs $\tau_{i}$
increase domestic production $Q_{i}^{dom}$ and the production inefficiency
$\frac{\delta}{2}\big(Q_{i}^{dom}-Q_{i}^{o}\big)^{2}$ increases quadratically
in the distortion of domestic production from the free trade level.
Thus if the goal is to increase $X_{i}$ sufficiently far above the
free trade level, it can become optimal to augment tariffs with subsidies.

It is often easier to reason intuitively about tariffs and export
subsidies, since they affect only domestic sales or exports, respectively,
than about production subsidies $s_{i}$, which influence both. The
following result will therefore be useful.
\begin{lem}
\label{lem:no_s-1}Consider two vectors of effective tariffs and subsidies
$(\tilde{\tau}_{i},\tilde{e}_{i},s_{i})$ and $(\tilde{\tau}_{i}',\tilde{e}_{i}',s_{i}')$.
If they satisfy
\begin{align}
\tilde{\tau}_{i}' & =\tilde{\tau}_{i}+s_{i}-s_{i}'\\
\tilde{e}_{i}' & =\tilde{e}_{i}+s_{i}-s_{i}'
\end{align}
then for any given effective tariffs and subsidies of country $j$
both vectors yield the same production quantities and direct costs
for both countries. In particular, this result also holds for $s_{i}'=0$.
\end{lem}
If countries have strategic production goals, welfare optimal production
levels can deviate from those that minimize total production cost.
The following result characterizes optimal production allocation conditional
that certain total production levels for both countries shall be achieved.
\begin{prop}
\label{prop:cond_eff}Let $\bar{D}^{o}(X_{A},X_{B})$ denote the minimum
total direct cost under the restriction that countries $A$ and $B$
produce $X_{A}$ and $X_{B}$, respectively. For any market outcome
that yields $X_{A}$ and $X_{B}$ with total direct excess costs $D$
define by $\bar{E}(X_{A},X_{B})=D-\bar{D}^{o}(X_{A},X_{B})$ the excess
cost compared to the lowest cost achievable when producing $X_{A}$
and $X_{B}$. We have 
\begin{equation}
\bar{E}(X_{A},X_{B})=\frac{\delta}{4}\,\big(Q_{i}^{dom}-Q_{i}^{exp}\big)^{2}\label{eq:condDW-quantities}
\end{equation}
for either country $i\in\{A,B\}$. In an interior solution with $Q_{i}^{dom}=\tilde{Q}_{i}^{dom}$
and $Q_{i}^{exp}=\tilde{Q}_{i}^{exp}$ for both countries, it also
satifies 
\end{prop}
\noindent 
\begin{equation}
\bar{E}(X_{A},X_{B})=\frac{1}{4\delta}\,\Big(\,(\tilde{\tau}_{A}+\tilde{\tau}_{B})-(\tilde{e}_{A}+\tilde{e}_{B})\,\Big)^{2}.\label{eq:condDW-policies}
\end{equation}

So irrespective of the strength of each country $i$'s concerns for
minimum production levels $X_{i}$, from a welfare point of view it
is always optimal if for each country $i$ domestic production is
equal to its exports $Q_{i}^{dom}=Q_{i}^{exp}$, which is the case
if for both countries effective export subsidies are equal to effective
tariffs. 

\section{Strategic Interactions and Stability of International Agreements}

We now analyze strategic interactions between both countries and the
stability of international trade agreements. For the sake of tractability,
we restrict attention to the following scenario: Country A has a competitive
disadvantage in the strategic sector and wants to implement a goal
for domestic production above the free trade equilibrium level. Formally,
the subsequent analysis satisfies the following assumption.
\begin{assumption}
\label{as:scen}Assume $\alpha_{A}\leq\alpha_{B}$, i.e. country $A$
is a net importer in the free-trade equilibrium. Further assume country
$A$ has a strategic production goal $\bar{X}_{A}$ satisfying $X_{A}^{o}<\bar{X}_{A}<1$
where $X_{A}^{o}$ is the free trade production of country $A$. Country
$A$'s utility function is given by
\begin{equation}
u_{A}=-\lambda_{A}\max(\bar{X}_{A}-X_{A},0)-D_{A}
\end{equation}
with $\lambda_{A}$ measuring the marginal disutility of falling behind
the strategic goal. Country $B$ simply prefers higher production
in the strategic sector without a specific goal: 
\begin{equation}
u_{B}=\gamma_{B}X_{B}-D_{B}.
\end{equation}
Assume $\lambda_{A}\gg\gamma_{B}>0$.
\end{assumption}

\subsection{Policy choice in a Nash equilibrium without tradeable import certificates}

The next result characterizes the Nash equilibrium if countries endogenously
and simultaneously choose their tariffs and export subsidies but don't
implement a TIC scheme. Given Lemma \ref{lem:no_s-1}, we simplify
the analysis by assuming w.l.o.g. that production subsidies $s_{A}$
and $s_{B}$ are zero.
\begin{prop}
\label{prop:eq_tic}Suppose Assumption 1 holds and that countries
simultaneously choose export subsidies and tariffs, with no TIC scheme
or production subsidies. In every Nash equilibrium: (i) country A
will produce its target $X_{A}=\bar{X}_{A}$, (ii) both countries
choose tariffs weakly larger than export subsidies, with strict inequality
for country $A$; (iii) the equilibrium is not conditionally efficient.
Moreover, in an interior solution there is a unique Nash equilibrium
characterized by 
\begin{align}
\tau_{B}^{*} & =\gamma_{B}\label{eq:tau_b_no_TIC}\\
e_{B}^{*} & =0\nonumber 
\end{align}
and 
\begin{align}
\tau_{A}^{*} & =-\alpha_{A}+\frac{2}{3}\,\delta\,\bar{X}_{A}+\frac{1}{3}\,\gamma_{B}\\
e_{A}^{*} & =-\alpha_{A}+\frac{1}{3}\,\delta\,\bar{X}_{A}+\frac{2}{3}\,\gamma_{B}
\end{align}
\end{prop}
That tariffs are set at higher levels than export subsidies is consistent
with the financial incentives identified in Proposition \ref{prop:E}.
This stronger reliance on tariffs over export subsidies reduces the
volume of international trade and yields an outcome that is not conditionally
efficient. Accordingly, there is scope for international agreements
that take country A\textquoteright s strategic concerns into account
while achieving a more efficient allocation.

\subsection{Enforceable international agreements}

In this section we consider the case that enforceable trade agreements
can be written. The following result characterizes an agreement with
a TIC that implements country A's strategic production target $\bar{X}_{A}$
in a conditionally efficient manner.
\begin{prop}
\label{prop:opt_tic} \textbf{TIC Agreement} Assume country $A$ implements
a TIC scheme with $\eta_{A}$ and $\phi_{A}$ chosen such that
\begin{align}
\eta_{A} & =\frac{2-\bar{X}_{A}}{\bar{X}_{A}}>1,\label{eq:eta_opt}\\
\phi_{A}\eta_{A} & =1.\label{eq:phi_opt}
\end{align}
Country B may also implement a TIC scheme with $\eta_{B}\geq1$.
No country uses direct tariffs or subsidies. The resulting market
outcome is conditionally efficient and satisfies $X_{A}=\bar{X}_{A}$
with
\begin{align}
\pi_{A}=\tilde{e}_{A}=\tilde{\tau}_{A}= & \alpha_{A}\bar{\chi}_{A}>0\\
\pi_{B}=\tilde{e}_{B}=\tilde{\tau}_{B} & =0
\end{align}
where
\begin{equation}
\bar{\chi}_{A}=\frac{\bar{X}_{A}-X_{A}^{o}}{X_{A}^{o}}
\end{equation}
measures the relative change of country A's strategic production target
$\bar{X}_{A}$ compared to its free trade production level $X_{A}^{o}$.
Both countries A and B have higher utility than in the Nash equilibria
characterized in Proposition \ref{prop:eq_tic}.
\end{prop}
A key feature of the agreement described above is that both countries
are better off than in the Nash equilibrium without an agreement.
However, many alternative agreements could also achieve Pareto improvements
of this kind. The main advantage of the proposed TIC Agreement is
its simplicity and symmetry: the policy instruments can be designed
in a structurally identical way for both countries, and the crucial
parameters $\eta_{i}$ and $\phi_{i}$ need not be calibrated to the
underlying production technologies, which are typically imperfectly
known and subject to change over time.

Negotiations that start from asymmetric positions are rarely straightforward.
Yet, countries might find it easier to agree on a general principle
governing the choice of $\eta_{i}$ and $\phi_{i}$. The parameter
$\eta_{i}$ effectively determines how strongly a country concerned
about its domestic production capacity may restrict imports relative
to exports. A principle that most countries would likely endorse is
that no country should set $\eta_{i}<1$, as this would require exports
to exceed imports and if symmetrically chosen would lead to a collapse
of trade. The range of $\eta_{i}\geq1$ values that countries find
mutually acceptable would depend on the perceived trade-off between
the benefits of more open trade and the strength of their strategic
concerns. In a world with high mutual trust among trading partners,
where strategic risks are viewed as limited, larger values of $\eta_{i}$
would likely be more feasible than in one characterized by low trust
and heightened strategic rivalry. 

Regarding $\phi_{i}$, country $A$ would gain from a modified agreement
that sets $\phi_{A}$ below $1/\eta_{A}$, i.e. granting a higher
share of the certificate income to the government and a lower share
to exporters. That is because countries prefer higher tariffs over
equivalent export subsidy payments, even though conditional efficiency
requires equal effective rates. A lower value of $\phi_{A}$ would,
in equilibrium, raise the effective tariff $\tilde{\tau}_{A}$ and
lower the effective export subsidy $\tilde{e}_{A}$, reducing $A$'s
direct costs at the expense of country $B$. However, countries might
agree on the principle that trade restrictions should not be accompanied
by additional fiscal gains for the country that derives strategic
benefits from those restrictions, whenever such gains come at the
cost of efficiency. This consideration provides an argument for conditionally
efficient agreements in which effective export subsidies and tariffs
are set to equal levels.

The same market outcome as under the TIC Agreement described above
could also be implemented through an agreement that relies solely
on export subsidies and tariffs:
\begin{cor}
\label{cor:opt_taue}\textbf{No-TIC Agreement} If country A uses no
TIC scheme and fixes $e_{A}=\tau_{A}=\alpha_{A}\bar{\chi}_{A}$ and
$s_{A}=0$ and country $B$ fixes $e_{B}=\tau_{B}=s_{B}=0$ the same
conditionally efficient market outcome as in Proposition \ref{prop:opt_tic}
is implemented.
\end{cor}
A key difficulty with the No-TIC Agreement is that tariff and subsidy
rates depend on production technologies. Since the relevant parameters
are likely to be imperfectly known and unevenly shared, asymmetric
information is a concern - an obstacle well known to undermine efficient
negotiation outcomes.\footnote{See e.g. Myerson and Satterthwaite (1983) for a seminal theoretic
contribution. For a recent contribution in the context of trade agreements,
see Yamamoto (2024).} Another challenge is that the relative production capacities and
technology will likely change over time, which would require adjustment
of the negotiated tariff and subsidy rates. Tariffs and subsidies
are often politically difficult to unwind once in place, even when
their original justification no longer holds, see e.g. Freund and
Özden (2008) or Magee (2002). In contrast, adjustments of effective
tariffs and subsidies are automated in the TIC Agreement and vanish
once the trade deficits becomes sufficiently small.

Even if negotiated tariffs or a formal TIC scheme could be reasonably
well enforced, restricting subsidies seems far more challenging in
practice, since subsidies may appear in various forms, like below-market
loans from public banks.\footnote{For example, Hillman and Manak (2023) summarize: ``The {[}WTO{]}
rules did little to prevent widespread industrial subsidies use by
countries hoping to gain an edge in international trade.''} The following result shows that even if the TIC Agreement cannot
legally enforce zero subsidies, it is able to substantially limit
the incentives and consequences for hidden subsidies. 
\begin{prop}
\label{prop:TIC_deviation_s} Consider the TIC Agreement of Proposition~\ref{prop:opt_tic}.
Now assume both countries could deviate from the agreement by choosing
arbitrary non-negative subsidies $e_{A},e_{B},s_{A},s_{B}\ge0$ while
all other aspects of the agreement remain legally enforced. Then:

(i) Country $A$ has no profitable deviation to any $(e_{A},s_{A})$
with $e_{A}>0$ or $s_{A}>0$.

(ii) If $\eta_{A}=1$, country $B$ has no profitable deviation to
any $(e_{B},s_{B})$ with $e_{B}>0$ or $s_{B}>0$. If $\eta_{A}>1$,
there is a threshold 
\begin{equation}
\gamma_{B}^{TIC}(\eta_{A})=\delta\cdot\frac{\eta_{A}^{2}+\eta_{A}-1}{\eta_{A}^{2}-1}
\end{equation}
such that $B$ has an incentive to raise its export subsidy if and
only if $\gamma_{B}>\gamma_{B}^{TIC}(\eta_{A})$. The threshold satisfies
$\lim_{\eta_{A}\to1}\gamma_{B}^{TIC}(\eta_{A})=\infty$.

(iii) For any $(e_{B},s_{B})$, country $A$'s production satisfies
the global bound $X_{A}\ge1/\eta_{A}$ and its direct costs $D_{A}$
don't increase in the other country's subsidies $e_{B}$ or $s_{B}$.
\end{prop}
Given our previous insights, it might seem surprising at first sight
that country $B$ could have a strategic interest at all to set positive
subsidies given a binding TIC scheme in country $A$. We first explore
the intuition for export subsidies $e_{B}>0$. The key mechanism is
as follows: positive export subsidies by country $B$ will increase
the certificate price $\pi_{A}$ which leads to both higher effective
tariffs $\tilde{\tau}_{A}$ and higher effective export subsidies
$\tilde{e}_{A}$ in country $A$. As effective tariffs $\tilde{\tau}_{A}$
increase less than $e_{B}$, the net effect is that domestic production
$Q_{A}^{dom}$ in $A$ will fall, while its exports $Q_{A}^{exp}$
increase. Each unit of additional exports allows $\eta_{A}$ additional
units of imports by country $B$. For $\eta_{A}>1$ country $A$'s
imports thus increase more than its exports. As result the induced
shift from domestic production to exports will decrease country A's
total output $X_{A}$ and correspondingly increase country $B$'s
total production. Thus if country $B$ has a large strategic preference
$\gamma_{B}$ for higher production it may choose export subsidies
due to this effect. The effect becomes smaller the closer $\eta_{A}$
is to 1 and completely vanishes for $\eta_{A}=1$. 

Looking at strictly positive production subsidies $s_{B}>0$, the
expansive effect on $B$'s imports that is described above is more
limited or non-existent, since positive production subsidies $s_{B}>0$
also reduce country $A$'s exports.

Countries may subsidize strategic-sector production for reasons unrelated
to gaining an advantage in international trade, for example to support
underdeveloped regions or to accelerate the transition to environmentally
friendly production technologies. Part (iii) of Proposition \ref{prop:TIC_deviation_s}
implies that a country with a TIC in place can be less concerned about
such subsidies from its trading partner. The maximum negative impact
on total production in the strategic sector is strictly limited and
the subsidies from the trading partner even bring a financial gain.

For a short quantitative example, consider the production target $\bar{X}_{A}=0.8$
such that the TIC-agreement from Proposition \ref{prop:opt_tic} specifies
$\eta_{A}=1.5$. Even if country B deviates by setting arbitrarily
large export subsidies $e_{B}$, country $A$ still has a guaranteed
total production of $X_{A}\geq\frac{1}{\eta_{A}}=\frac{2}{3}$ which
is 83.3\% of the target $\bar{X}_{A}$. If even such a limited shortfall
from the target $\bar{X}_{A}$ is considered very problematic, a possible
counter-measure is to pick a somewhat lower value for $\eta_{A}$
than in an agreement with enforceable subsidies. In our example, with
$\bar{X}_{A}=0.8$ setting $\eta_{A}=1.25$ guarantees $X_{A}\geq\bar{X}_{A}$
even if country $B$ sets arbitrarily high export subsidies.

Under the No-TIC Agreement, subsidies implemented by country $B$
directly stimulate its output without triggering a countervailing
increase in certificate prices. Consequently, subsidies become attractive
already for substantially lower values of the strategic production
preference $\gamma_{B}$ and there is no guaranteed minimum level
of production in the strategic sector.
\begin{prop}
\label{prop:NO_TIC_threshold_compact} Consider the No-TIC Agreement
of Corollary~\ref{cor:opt_taue} and assume country $B$ can deviate
by setting positive subsidies $s_{B}>0$ or $e_{B}>0$ while tariffs
remain fixed at $\tau_{B}=0$. Then country $B$ has strict incentives
to choose positive subsidies if 
\begin{equation}
\gamma_{B}>\frac{1}{2}\,\delta\,\frac{\eta_{A}}{1+\eta_{A}}\equiv\gamma_{B}^{No\text{-}TIC}
\end{equation}
with $\eta_{A}=\frac{2-\bar{X}_{A}}{\bar{X}_{A}}$. The ratio of the
respective thresholds for the TIC Agreement and No-TIC Agreement is
given by 
\begin{equation}
\frac{\gamma_{B}^{TIC}}{\gamma_{B}^{No\text{-}TIC}}=2\,\frac{\eta_{A}^{2}+\eta_{A}-1}{\eta_{A}(\eta_{A}-1)}>2.
\end{equation}
\end{prop}
For example, for $\eta_{A}=1.5$ the threshold $\gamma_{B}^{TIC}(\eta_{A})$
for $\gamma_{B}$ above which country $B$ deviates from the TIC Agreement
with positive subsidies is 7 times higher than the corresponding threshold
$\gamma_{B}^{No-TIC}$ for the No-TIC Agreement.

Beyond subsidies, various forms of non-tariff trade barriers, such
as certain regulatory measures, are often difficult to verify. The
following result shows that the TIC Agreement effectively eliminates
incentives for such non-tariff barriers, whereas these incentives
may persist under the No-TIC Agreement.
\begin{prop}
\label{prop:ntb}Assume each country can decide a level of non-tariff
trade-barriers $\beta_{i}\geq0$ that increases the marginal costs
of imports from country $j$ by $\beta_{i}$ but unlike a tariff grants
no tariff income for country $i$. Tariffs and subsidies are enforceable
as specified in the agreement but non-tariff barriers can be freely
set. Then with the TIC Agreement it is strictly optimal for both countries
to have zero non-tariff trade barriers. In contrast, under the No-TIC
Agreement country $B$ has incentives to erect non-tariff trade barriers
$\beta_{B}>0$ if 
\begin{equation}
\gamma_{B}>\frac{\delta}{1+\eta_{A}}
\end{equation}
with $\eta_{A}=\frac{2-\bar{X}_{A}}{\bar{X}_{A}}$.
\end{prop}

\subsection{Market Power}

Our analysis above indicates that the TIC Agreement performs well
both in terms of incentives and efficiency. A key drawback, however,
is that TIC schemes can induce socially costly adverse behavior when
the strategic sector is sufficiently small or concentrated for domestic
producers to exercise market power in the certificate market.

To illustrate, first consider an extreme case. Suppose there is a
single monopolistic producer in country A that covers the entire strategic
sector. By choosing its export quantity, the monopolist can arbitrarily
limit the number of certificates and import volume into its domestic
market. For instance, by refraining from exporting altogether, the
monopolist could completely exclude foreign competition.

For a model with oligopolistic market power in the certificate market
assume that country $B$ keeps perfectly competitive producers as
before while there are $N$ large producers in country $A$ indexed
by $n=1,...,N$. For every market $m\in[0,1]$ there is exactly one
producer $n_{m}$ from country A that can serve it and for every segment
$[m_{1},m_{2}]$ of the market line the mass of markets that can be
served by each producer $n=1,...,N$ shall be equal to $(m_{2}-m_{1})/N$.
Thus the producers from A are evenly distributed across the markets.
If in market $m$ in country $i$ the large producer $n_{m}$ and
the competitive producers from country $B$ set the same price, the
market shall be completely served by the large producer. Marginal
costs remain as specified in Section 2.

For each market $m$ in country $i$ we can restrict attention to
one of two strategies for the large producer $n_{m}$: either it captures
the market at the marginal cost of producers from country $B$ or
it stays out of the market. Given that the cost advantage for country
$A$ producers declines with $m$, an optimal strategy for every large
producer $n$ can be described by two variables $m_{n}^{dom}$ and
$m_{n}^{exp}$ meaning that it will capture all domestic markets in
country A in the interval $[0,m_{n}^{dom}]$ and all export markets
in country $B$ in the interval $[0,m_{n}^{exp}]$. We can also describe
these strategies by the corresponding production quantities
\begin{equation}
q_{n}^{dom}=\frac{1}{N}m_{n}^{dom},\qquad q_{n}^{exp}=\frac{1}{N}m_{n}^{exp}
\end{equation}

Total production quantities for country $A$ are correspondingly 
\begin{equation}
Q_{A}^{dom}=\sum_{n=1}^{N}q_{n}^{dom},\qquad Q_{A}^{exp}=\sum_{n=1}^{N}q_{n}^{exp}
\end{equation}
If producer $n$ captures a domestic market $m$, it earns in that
market a profit of 
\begin{equation}
g_{n}^{dom}(m)=c_{B}^{exp}(m)-c_{A}^{dom}(m)\label{eq:g_dom}
\end{equation}

and if it captures export market $m$ it earns in that market a profit
of
\begin{equation}
g_{n}^{exp}(m)=c_{B}^{dom}(m)-c_{A}^{exp}(m)\label{eq:g_exp}
\end{equation}

We have set up this model such that market power in product markets
has no effect on countries' total production levels $Q_{i}^{dom}$
and $Q_{i}^{exp}$ as long as there are no certificate prices that
could be strategically manipulated by large producers.
\begin{cor}
If no country implements a TIC scheme, then independent of the number
of large producers $N$, countries A and B will serve exactly the
same markets as if all producers were perfectly competitive. The only
difference to the market outcome under perfect competition is that
in the markets served by a large producer from country $A$ prices
are determined by the marginal cost of the competitive producers from
country $B$ instead of the marginal cost from country A producers.
\end{cor}
Another direct implication is
\begin{cor}
The No-TIC Agreement from Corollary \ref{cor:opt_taue} remains conditionally
efficient in our model with $N$ large producers in country A.
\end{cor}
We now analyse the market outcome under the TIC Agreement from Proposition
\ref{prop:opt_tic}.

The timing of the model shall be as follows. First, every strategic
producer $n$ simultaneously chooses its export quantity $q_{n}^{exp}\geq0$.
In this step, each large producer takes into account how its export
quantity will affect the certificate price $\pi_{A}$. In the second
step all other quantities are determined such that every firm maximizes
its profits taking equilibrium certificate prices $\pi_{A}$ as given.
\begin{prop}
\label{prop:oligopoly}Suppose the TIC Agreement of Proposition~\ref{prop:opt_tic}
is in place and country A has $N$ large producers. In the unique
and symmetric Nash equilibrium, we have
\begin{equation}
Q_{A}^{\exp}\;=\;\begin{cases}
{\displaystyle \frac{N-\eta_{A}}{\,N(\eta_{A}+1)-\eta_{A}(\eta_{A}-1)\,}\,,} & \text{if }N>\eta_{A},\\[8pt]
0\,, & \text{if }N\le\eta_{A},
\end{cases}\qquad\pi_{A}\;=\;\delta\bigl(1-\eta_{A}Q_{A}^{\exp}\bigr)-\alpha_{A}.
\end{equation}
In all cases $Q_{A}^{\exp}<Q_{A}^{\text{dom}}$ and the production
allocation is not conditionally efficient. Moreover, for $N>\eta_{A}$,
$Q_{A}^{\exp}$ is strictly increasing in $N$ and $\pi_{A}$ is strictly
decreasing in $N$. As $N\to\infty$ the production quantities and
certificate price converge to the perfectly competitive outcome stated
in Proposition~\ref{prop:opt_tic}:
\begin{equation}
Q_{A}^{\exp}\to\frac{1}{1+\eta_{A}},\qquad Q_{A}^{\text{dom}}\to\frac{1}{1+\eta_{A}},\qquad\pi_{A}\to\frac{\delta}{1+\eta_{A}}-\alpha_{A}.
\end{equation}
\end{prop}
For example, in the special case $\eta_{A}=1$ the export quantities
satisfy $Q_{A}^{exp}=\frac{N-1}{2N}$. Then for $N=2$ the resulting
exports in country $A$ are $Q_{A}^{exp}=\frac{1}{4}$, which are
50\% below the competitive export level of $\frac{1}{2}$. To halve
the gap to the competitive export level, the number of firms must
be doubled: for $N=4$ exports are 25\% below the competitive level,
for $N=8$, the gap is 12.5\%, and so on.

Of course, our stylized oligopoly model is not meant to deliver a
reliable quantitative assessment of the distortions created by market
power in the certificate market. Its value is instead illustrative:
it can help build a first-pass intuition for the effects at work.

\section{Concluding Remarks}

Our results show that, for competitive sectors, TIC-based trade agreements
can align strategic minimum-production objectives with conditionally
efficient allocations of domestic production and exports among trading
partners. At the same time, they can effectively dampen countries'
incentives to use opaque subsidies or non-tariff barriers. Important
implementation questions remain, however. Certificate duration (single-period
versus bankable), inventory limits, and auction design will shape
manipulation risks and price volatility. Design lessons from pollution-permit
markets are relevant here; see, for example, Liski and Montero (2011)
on the role of permit storability and Hahn (1984) for an early analysis
of market power in markets for pollution certificates.

Alternatively, one could replicate the automatic adjustment with explicitly
state-contingent tariffs and subsidies. For instance, specify formula-based
adjustments that respond to observed shortfalls from a declared export-to-import
ratio, combined with pre-announced review points and high transparency
to support planning and mutual monitoring.

Even if tradeable import certificates were replaced by such automatically
adjusted tariffs and export subsidies, adverse incentives for large
domestic producers would persist. By withholding exports, domestic
firms can raise the effective tariff protection at home and the export
subsidy rate. Unless such strategic behavior by firms can be effectively
limited, fully rule-based tariff and subsidy adjustments or TIC schemes
carry substantial risks.

Finally, any real-world implementation requires a multi-partner lens.
One would likely condition effective tariff and subsidy rates by partner-specific
supply risk and exposure paired with explicit diversification commitments.
Well-designed schemes that explicitly allow for friendshoring, or
treat sufficient diversification as a substitute to domestic production,
may also help mitigate risks arising from domestic market power. We
consider a detailed analysis of such a multilateral mechanism an important
avenue for future research.

\section*{References}

\begingroup
\setlength{\parindent}{0pt}
\setlength{\parskip}{0.8\baselineskip}
\hangindent=1.5em
\hangafter=1
\raggedright

Aghion, P., Cai, J., Dewatripont, M., Du, L., Harrison, A., \& Legros,
P. (2015). Industrial policy and competition. \textit{American Economic
Journal: Macroeconomics, 7}(4), 1--32. https://doi.org/10.1257/mac.20120103

Amiti, M., Redding, S. J., \& Weinstein, D. E. (2019). The impact
of the 2018 tariffs on prices and welfare. \textit{Journal of Economic
Perspectives, 33}(4), 187--210. https://doi.org/10.1257/jep.33.4.187

Bernard, A. B., Redding, S. J., \& Schott, P. K. (2010). Multiple-product
firms and product switching. \textit{American Economic Review, 100}(1),
70--97. https://doi.org/10.1257/aer.100.1.70

Bhagwati, J., \& Srinivasan, T. N. (1969). Optimal intervention to
achieve non-economic objectives. \textit{The Review of Economic Studies,
36}(1), 27--38. https://doi.org/10.2307/2296340

Buffett, W. E. (2003, November 10). America's growing trade deficit
is selling the nation out from under us. Here's a way to fix the problem
-- and we need to do it now. \textit{Fortune}.

Congressional Budget Office (CBO). (2025). \textit{Budgetary and economic
effects of increases in tariffs implemented between January 6 and
May 13, 2025}. Washington, DC: Congressional Budget Office. Retrieved
November 13, 2025, from https://www.cbo.gov/publication/61389

Dawar, K. (2020). Regulating under the radar: EU official export credit
support. \textit{VoxEU/CEPR.} https://cepr.org/voxeu/columns/regulating-under-radar-eu-official-export-credit-support-0

Dornbusch, R., Fischer, S., \& Samuelson, P. A. (1977). Comparative
advantage, trade, and payments in a Ricardian model with a continuum
of goods. \textit{American Economic Review, 67}(5), 823--839.

European Commission (2021): The European economic and financial system:
fostering openness, strength and resilience, COM(2021) 32 final, 19
January 2021.

European Commission. (2023). \textit{A Green Deal Industrial Plan
for the Net-Zero Age} (COM(2023) 62). https://eur-lex.europa.eu/legal-content/EN/TXT/?uri=celex:52023DC0062

Freund, C., \& Özden, C. (2008). Trade policy and loss aversion. \textit{American
Economic Review, 98}(4), 1675--1691. https://doi.org/10.1257/aer.98.4.1675

Goldlücke, S., \& Kranz, S. (2023). Reconciling relational contracting
and hold-up: A model of repeated negotiations. \textit{Journal of
the European Economic Association, 21}(3), 864--906. https://doi.org/10.1093/jeea/jvac036

Hahn, R. W. (1984). Market power and transferable property rights.
\textit{The Quarterly Journal of Economics, 99}(4), 753--765. https://doi.org/10.2307/1883124

Handley, K., \& Limao, N. (2017). Policy uncertainty, trade, and welfare:
Theory and evidence for China and the United States. \textit{American
Economic Review, 107}(9), 2731--2783. https://doi.org/10.1257/aer.20141419

Hillman, J., \& Manak, I. (2023). \textit{Rethinking international
rules on subsidies} (Council Special Report No. 96). Council on Foreign
Relations.

Juhász, R., Lane, N., \& Rodrik, D. (2024). The new economics of industrial
policy. \textit{Annual Review of Economics}, 16(1), 213--242. https://doi.org/10.1146/annurev-economics-081023-024638

Krueger, A. O. (1974). The political economy of the rent-seeking society.
\textit{American Economic Review, 64}(3), 291--303.

Liski, M., \& Montero, J. P. (2011). Market power in an exhaustible
resource market: The case of storable pollution permits. \textit{The
Economic Journal, 121}(551), 116--144. https://doi.org/10.1111/j.1468-0297.2010.02366.x

Magee, C. (2002). Declining industries and persistent tariff protection.
\textit{Review of International Economics, 10}(4), 749--762. https://doi.org/10.1111/1467-9396.00362

Martin, P., Mayer, T., \& Thoenig, M. (2008). Make trade not war?
\textit{Review of Economic Studies, 75}(3), 865--900. https://doi.org/10.1111/j.1467-937X.2008.00492.x

McBride, J., \& Chatzky, A. (2019, May 13). Is \textquoteleft Made
in China 2025\textquoteright{} a threat to global trade? Council on
Foreign Relations. https://www.cfr.org/backgrounder/made-china-2025-threat-global-trade

Melvin, J. R. (1986). The nonequivalence of tariffs and import quotas.
\textit{American Economic Review}, 76(5), 1131--1134

Myerson, R. B., \& Satterthwaite, M. A. (1983). Efficient mechanisms
for bilateral trading. \textit{Journal of Economic Theory, 29}(2),
265--281. https://doi.org/10.1016/0022-0531(83)90048-0

National Institute of Standards and Technology (NIST). (2023). \textit{CHIPS
incentives funding opportunities}. Retrieved November 13, 2025, from
https://www.nist.gov/chips/chips-incentives-funding-opportunities

Papadimitriou, D. B., Hannsgen, G., \& Zezza, G. (2008). The Buffett
plan for reducing the trade deficit. \textit{Levy Economics Institute
Working Paper}.

Pavcnik, N. (2002). Trade liberalization, exit, and productivity improvements:
Evidence from Chilean plants. \textit{Review of Economic Studies,
69}(1), 245--276. https://doi.org/10.1111/1467-937X.00205

Rhode, P. W., Snyder, J. M., Jr., \& Strumpf, K. (2018). The arsenal
of democracy: Production and politics during WWII. \textit{Journal
of Public Economics, 166}, 145--161. https://doi.org/10.1016/j.jpubeco.2018.08.010

Shibata, H. (1968). A note on the equivalence of tariffs and quotas.
\textit{American Economic Review, 58}(1), 137--142.

State Council of the People's Republic of China. (2015). \textit{Made
in China 2025}. (State Council plan for upgrading Chinas manufacturing
sector.) English translation available via the Center for Security
and Emerging Technology. Retrieved November 13, 2025, from https://cset.georgetown.edu/wp-content/uploads/t0432\_made\_in\_china\_2025\_EN.pdf 

Studwell, J. (2013). \textit{How Asia works: Success and failure in
the world's most dynamic region}. Grove Press.

Tower, E. (1975). The optimum quota and retaliation. \textit{Review
of Economic Studies}, 42(4), 623--630. https://doi.org/10.2307/2296799

Trefler, D. (2004). The long and short of the Canada-U.S. Free Trade
Agreement. \textit{American Economic Review, 94}(4), 870--895. https://doi.org/10.1257/0002828042002633

White House. (2017, July 21). Presidential Executive Order 13806:
Assessing and strengthening the manufacturing and defense industrial
base and supply chain resiliency of the United States. \textit{The
White House}.

Yamamoto, K. (2024). Free Trade Agreements Under Asymmetric Information
and the Negotiating Role of Side Payments. Available at SSRN 4716997.

\endgroup

\section*{}

\section*{Appendix}

\subsection*{A1 Proofs}

\subsubsection*{Proof for Proposition \ref{prop:market_eq}}

Country $i$ serves $j$'s market if $c_{i}^{exp}(m)<c_{j}^{dom}(m)$,
i.e. 
\[
w_{i}(m)-s_{i}-\tilde{e}_{i}+\tilde{\tau}_{j}\;<\;w_{j}(m)-s_{j}.
\]
Rearranging yields the export cutoff condition in terms of the technology
gap $\Delta w_{i}(m)\equiv w_{i}(m)-w_{j}(m)$: 
\[
\Delta w_{i}(m)<(s_{i}-s_{j})+(\tilde{e}_{i}-\tilde{\tau}_{j}).
\]
Using $\Delta w_{A}(m)=\delta(m-m^{o})$ with $m^{o}=Q_{A}^{o}=\alpha_{A}/\delta$,
the export cutoff is 
\[
m_{i}^{exp}=m^{o}+\frac{(s_{i}-s_{j})+(\tilde{e}_{i}-\tilde{\tau}_{j})}{\delta}.
\]

Similarly, $i$ serves its own domestic market when $c_{i}^{dom}(m)\le c_{j}^{exp}(m)$:
\[
w_{i}(m)-s_{i}\;\le\;w_{j}(m)-s_{j}-\tilde{e}_{j}+\tilde{\tau}_{i}\quad\Longleftrightarrow\quad\Delta w_{i}(m)\le(s_{i}-s_{j})+(\tilde{\tau}_{i}-\tilde{e}_{j}).
\]
Hence the domestic cutoff is 
\[
m_{i}^{dom}=m^{o}+\frac{(s_{i}-s_{j})+(\tilde{\tau}_{i}-\tilde{e}_{j})}{\delta}.
\]

Because each product market has unit mass and the weak inequalities
bind only on a null set, the corresponding quantities equal these
cutoffs, truncated to $[0,1]$: 
\[
Q_{i}^{dom}=\max(0,\min(1,m_{i}^{dom})),\qquad Q_{i}^{exp}=\max(0,\min(1,m_{i}^{exp})).
\]
Recognizing $Q_{A}^{o}=m^{o}$ and $Q_{B}^{o}=1-m^{o}$ yields the
stated expressions for $\tilde{Q}_{i}^{dom}$ and $\tilde{Q}_{i}^{exp}$.$\hfill\square$

\subsubsection*{Proof for Proposition \ref{prop:E}}

Using $Q_{i}^{imp}=Q_{j}^{exp}$ direct costs of country $i$ satisfy
\[
D_{i}=\int_{0}^{1}\big[w_{i}(m)\,q_{i}^{dom}(m)+w_{j}(m)\,q_{j}^{exp}(m)\big]\,dm+(s_{i}+\tilde{e}_{i})\,Q_{i}^{exp}-\left(s_{j}+\tilde{e}_{j}\right)Q_{i}^{imp}.
\]
All variables for free trade outcome will be superscripted with $^{o}$.
The direct costs under free trade are 
\[
D_{i}^{o}=\int_{0}^{1}\big[w_{i}(m)\,q_{i}^{dom,o}(m)+w_{j}(m)\,q_{j}^{exp,o}(m)\big]\,dm.
\]

Subtracting, yields
\begin{align*}
D_{i}-D_{i}^{o} & =\int_{0}^{1}\Big[w_{i}(m)\big(q_{i}^{dom}(m)-q_{i}^{dom,o}(m)\big)+w_{j}(m)\big(q_{j}^{exp}(m)-q_{j}^{exp,o}(m)\big)\Big]\,dm\\
 & \qquad+(s_{i}+\tilde{e}_{i})\,Q_{i}^{exp}-(s_{j}+\tilde{e}_{j})\,Q_{i}^{imp}.
\end{align*}
Since for each $m$ exactly one of $q_{i}^{dom}(m)$ and $q_{j}^{exp}(m)$
equals $1$, the integral term only loads on the set of products $m$
where production allocation differs from the free trade allocation.
We call it the missallocated set. Let $\Delta Q_{i}\equiv Q_{i}^{dom}-Q_{i}^{o}$.
By construction, the missallocated set is an interval of measure $|\Delta Q_{i}|$,
obtained by shifting the free trade cutoff $m^{o}$ left or right.
Using the linear technology gap $\Delta w_{i}(m)\equiv w_{i}(m)-w_{j}(m)=\delta(m-m^{o})$,
the incremental cost from misallocating the marginal interval of length
$|\Delta Q_{i}|$ is the area under a linear ramp: 
\[
\int_{0}^{|\Delta Q_{i}|}\delta x\,dx=\frac{\delta}{2}\,(\Delta Q_{i})^{2}.
\]
Therefore, 
\[
\int_{0}^{1}\Big[w_{i}(m)\big(q_{i}^{dom}-q_{i}^{dom,o}\big)+w_{j}(m)\big(q_{j}^{exp}-q_{j}^{exp,o}\big)\Big]\,dm=\frac{\delta}{2}\big(Q_{i}^{dom}-Q_{i}^{o}\big)^{2}.
\]
Putting all pieces together, we get 
\[
E_{i}=D_{i}-D_{i}^{o}=\frac{\delta}{2}\big(Q_{i}^{dom}-Q_{i}^{o}\big)^{2}+(s_{i}+\tilde{e}_{i})\,Q_{i}^{exp}-(s_{j}+\tilde{e}_{j})\,Q_{i}^{imp}.
\]

Summing up yields the first line of (\ref{eq:eq_E})
\[
E=E_{A}+E_{B}=\frac{\delta}{2}\left[\big(Q_{A}^{dom}-Q_{A}^{o}\big)^{2}+\big(Q_{B}^{dom}-Q_{B}^{o}\big)^{2}\right]
\]

Recognizing that
\begin{align*}
Q_{i}^{dom} & =1-Q_{j}^{exp}\\
Q_{i}^{o} & =1-Q_{j}^{o}
\end{align*}
yields the 2nd and 3rd lines by simple substitution and simplification.$\hfill\square$

\subsubsection*{Proof for Lemma \ref{lem:pi_binding-1}}

Because we are in an interior solution, Proposition~\ref{prop:market_eq}
implies for each country $k\in\{i,j\}$ that 
\[
Q_{k}^{exp}=Q_{k}^{o}+\frac{(s_{k}-s_{\ell})+(\tilde{e}_{k}-\tilde{\tau}_{\ell})}{\delta},
\]
where $\ell$ denotes the trading partner of $k$, $Q_{k}^{o}=\alpha_{k}/\delta$,
and $\delta=\alpha_{A}+\alpha_{B}$. Using this with $k=i$ and $\ell=j$
gives 
\begin{equation}
Q_{i}^{exp}=Q_{i}^{o}+\frac{(s_{i}-s_{j})+(\tilde{e}_{i}-\tilde{\tau}_{j})}{\delta}=\frac{\alpha_{i}}{\delta}+\frac{(s_{i}-s_{j})+(\tilde{e}_{i}-\tilde{\tau}_{j})}{\delta}.\label{eq:Qi_exp_raw}
\end{equation}
Similarly, with $k=j$ and $\ell=i$ we obtain 
\begin{equation}
Q_{j}^{exp}=Q_{j}^{o}+\frac{(s_{j}-s_{i})+(\tilde{e}_{j}-\tilde{\tau}_{i})}{\delta}=\frac{\alpha_{j}}{\delta}+\frac{(s_{j}-s_{i})+(\tilde{e}_{j}-\tilde{\tau}_{i})}{\delta}.\label{eq:Qj_exp_raw}
\end{equation}

Next, we substitute the effective instruments. By definition, 
\[
\tilde{\tau}_{k}=\tau_{k}+\pi_{k},\qquad\tilde{e}_{k}=e_{k}+\phi_{k}\eta_{k}\pi_{k}.
\]
We assume that country $j$ has no TIC or a non-binding TIC, so $\pi_{j}=0$.
Hence 
\[
\tilde{\tau}_{j}=\tau_{j},\qquad\tilde{e}_{j}=e_{j},
\]
while for country $i$ we have 
\[
\tilde{\tau}_{i}=\tau_{i}+\pi_{i},\qquad\tilde{e}_{i}=e_{i}+\phi_{i}\eta_{i}\pi_{i}.
\]
Substituting into \eqref{eq:Qi_exp_raw} and \eqref{eq:Qj_exp_raw}
yields 
\begin{equation}
Q_{i}^{exp}=\frac{\alpha_{i}}{\delta}+\frac{(s_{i}-s_{j})+\big(e_{i}+\phi_{i}\eta_{i}\pi_{i}-\tau_{j}\big)}{\delta},\label{eq:Qi_exp_pi}
\end{equation}
\begin{equation}
Q_{j}^{exp}=\frac{\alpha_{j}}{\delta}+\frac{(s_{j}-s_{i})+\big(e_{j}-(\tau_{i}+\pi_{i})\big)}{\delta}.\label{eq:Qj_exp_pi}
\end{equation}

A binding TIC equilibrium in country $i$ is characterized by $\pi_{i}>0$
and 
\[
\eta_{i}Q_{i}^{exp}=Q_{i}^{imp}>0.
\]
Since $Q_{i}^{imp}=Q_{j}^{exp}$ by definition, the binding condition
can be written as 
\begin{equation}
\eta_{i}Q_{i}^{exp}=Q_{j}^{exp}.\label{eq:TIC_binding_condition}
\end{equation}
Using \eqref{eq:Qi_exp_pi} and \eqref{eq:Qj_exp_pi} in \eqref{eq:TIC_binding_condition}
and multiplying both sides by $\delta$ gives 
\[
\eta_{i}\left[\alpha_{i}+(s_{i}-s_{j})+e_{i}+\phi_{i}\eta_{i}\pi_{i}-\tau_{j}\right]=\alpha_{j}+(s_{j}-s_{i})+e_{j}-\tau_{i}-\pi_{i}.
\]

We now solve this equation for $\pi_{i}$. Expanding the left-hand
side and bringing all terms to the right-hand side yields 
\[
0=-\alpha_{j}-e_{j}+\eta_{i}\big(\alpha_{i}+e_{i}+s_{i}-s_{j}-\tau_{j}\big)+\eta_{i}^{2}\phi_{i}\pi_{i}+\pi_{i}+s_{i}-s_{j}+\tau_{i}.
\]
Collecting the terms in $\pi_{i}$ on the left and all remaining terms
on the right gives 
\[
(1+\phi_{i}\eta_{i}^{2})\pi_{i}=\alpha_{j}-\eta_{i}\alpha_{i}+(1+\eta_{i})(s_{j}-s_{i})+(e_{j}-\eta_{i}e_{i})+(\eta_{i}\tau_{j}-\tau_{i}).
\]
Dividing by $1+\phi_{i}\eta_{i}^{2}>0$ yields (\eqref{eq:pi_bind-1}).$\hfill\square$

\subsubsection*{Proof for Lemma \ref{lem:no_s-1}}

Whether market $m$ in country $i$ will be domestically served depends
only on the cost difference
\[
c_{j}^{exp}(m)-c_{i}^{dom}(m)=w_{j}(m)-s_{j}-\tilde{e}_{j}+\tilde{\tau}_{i}-w_{i}(m)+s_{i}.
\]
This cost difference stays the same whenever $\tilde{\tau}_{i}'+s_{i}'=\tilde{\tau}_{i}+s_{i}$.
With $\tilde{\tau}_{i}'=\tilde{\tau}_{i}+s_{i}-s_{i}'$ this always
holds true. Similarly, Whether market $m$ in country $j$ will be
domestically served depends only on the cost difference
\[
c_{i}^{exp}(m)-c_{j}^{dom}(m)=w_{i}(m)-s_{i}-\tilde{e}_{i}+\tilde{\tau}_{j}-w_{j}(m)+s_{j}.
\]
This cost difference stays the same whenever $\tilde{e}_{i}'+s_{i}'=\tilde{e}_{i}+s_{i}$.
With $\tilde{e}_{i}'=\tilde{e}_{i}+s_{i}-s_{i}'$ this condition is
always satisfied.

Direct costs $d_{i}^{exp}(m),d_{i}^{imp}(m)$, $d_{i}^{dom}(m)$ and
$d_{j}^{exp}(m),d_{j}^{imp}(m)$, $d_{j}^{dom}(m)$ in every market
$m$ are independent of the tariff rate and depends only on the sum
of production and effective export subsidies $\tilde{e}_{i}+s_{i}=\tilde{e}_{i}'+s_{i}'$.$\hfill\square$

\subsubsection*{Proof for Proposition \ref{prop:cond_eff}}

Recall from Proposition~\ref{prop:E}, that the (unconditional) excess
cost relative to the globally efficient free trade outcome can be
written as 
\[
E=\frac{\delta}{2}\left[\big(Q_{i}^{dom}-Q_{i}^{o}\big)^{2}+\big(Q_{i}^{exp}-Q_{i}^{o}\big)^{2}\right]
\]
using either country $i$. Fix the totals $(X_{A},X_{B})$ with $X_{A}+X_{B}=2$.
The minimum cost $\bar{D}^{o}(X_{A},X_{B})$ subject to $X_{i}=Q_{i}^{dom}+Q_{i}^{exp}$
is obtained by minimizing 
\[
f(Q_{i}^{dom},Q_{i}^{exp})=(Q_{i}^{dom}-Q_{i}^{o})^{2}+(Q_{i}^{exp}-Q_{i}^{o})^{2}\quad\text{subject to}\quad Q_{i}^{dom}+Q_{i}^{exp}=X_{i}.
\]
Because $X_{i}\in[0,2]$, the unconstrained minimizer $Q_{i}^{dom}=Q_{i}^{exp}=X_{i}/2$
lies in $[0,1]$, so bounds do not bind. Substituting $Q_{i}^{dom}=X_{i}/2$
and $Q_{i}^{exp}=X_{i}/2$ yields 
\[
\min f=\frac{1}{2}\,(X_{i}-2Q_{i}^{o})^{2},
\]
and hence the constrained minimum (relative to $W^{o}$) equals 
\[
\bar{D}^{o}(X_{A},X_{B})-D^{o}=\frac{\delta}{4}\,(X_{i}-2Q_{i}^{o})^{2}.
\]
Therefore, the conditional excess cost beyond this constrained minimum
is 
\[
\bar{E}(X_{A},X_{B})=E-\big(D^{o}(X_{A},X_{B})-D^{o}\big)=\frac{\delta}{2}\Big[(d)^{2}+(e)^{2}\Big]-\frac{\delta}{4}(d+e)^{2},
\]
where $d\equiv Q_{i}^{dom}-Q_{i}^{o}$ and $e\equiv Q_{i}^{exp}-Q_{i}^{o}$.
Using $d^{2}+e^{2}-\tfrac{1}{2}(d+e)^{2}=\tfrac{1}{2}(d-e)^{2}$ gives
\[
\bar{E}(X_{A},X_{B})=\frac{\delta}{4}\,(Q_{i}^{dom}-Q_{i}^{exp})^{2},
\]
which proves \eqref{eq:condDW-quantities}. For the interior-cutoff
case, substitute $Q_{i}^{dom}-Q_{i}^{exp}=\big(\tilde{\tau}_{i}-\tilde{e}_{j}\big)-\big(\tilde{e}_{i}-\tilde{\tau}_{j}\big)$
divided by $\delta$, i.e., 
\[
Q_{i}^{dom}-Q_{i}^{exp}=\frac{(\tilde{\tau}_{i}+\tilde{\tau}_{j})-(\tilde{e}_{i}+\tilde{e}_{j})}{\delta},
\]
into \eqref{eq:condDW-quantities} to obtain \eqref{eq:condDW-policies}.$\hfill\square$

\subsubsection*{Proof for Proposition \ref{prop:eq_tic}}

Throughout this proof we first work in the interior region where the
cutoff formulas in Proposition~\ref{prop:market_eq} apply without
truncation and then indicate how to adjust for corner solutions. By
Lemma~\ref{lem:no_s-1} we can, without loss of generality, restrict
attention to $s_{i}=0$ for both countries and, since there is no
TIC scheme, to $\pi_{i}=0$, so that $\tilde{\tau}_{i}=\tau_{i}$
and $\tilde{e}_{i}=e_{i}$.

\emph{Step 1: Quantities and direct-cost derivatives.} From Proposition~\ref{prop:market_eq},
in an interior outcome we have 
\[
Q_{i}^{dom}=Q_{i}^{o}+\frac{\tau_{i}-e_{j}}{\delta},\qquad Q_{i}^{exp}=Q_{i}^{o}+\frac{e_{i}-\tau_{j}}{\delta},
\]
so total production in $i$ is 
\begin{equation}
X_{i}=Q_{i}^{dom}+Q_{i}^{exp}=2Q_{i}^{o}+\frac{\tau_{i}+e_{i}-(\tau_{j}+e_{j})}{\delta}.\label{eq:Xi_noTIC}
\end{equation}
Using Proposition~\ref{prop:E} with $s_{i}=0$ and $\tilde{e}_{i}=e_{i}$,
\[
D_{i}=D_{i}^{o}+E_{i},\qquad E_{i}=\frac{\delta}{2}\big(Q_{i}^{dom}-Q_{i}^{o}\big)^{2}+e_{i}\,Q_{i}^{exp}-e_{j}\,Q_{i}^{imp},
\]
and $Q_{i}^{imp}=Q_{j}^{exp}$. Substituting the interior expressions
for $Q_{i}^{dom},Q_{i}^{exp}$ and simplifying yields a quadratic
polynomial in the instruments. Differentiating with respect to own
instruments gives, for each country $i$, 
\begin{equation}
\frac{\partial D_{i}}{\partial\tau_{i}}=\frac{\tau_{i}}{\delta},\qquad\frac{\partial D_{i}}{\partial e_{i}}=Q_{i}^{o}+\frac{2e_{i}-\tau_{j}}{\delta}.\label{eq:Di_derivs}
\end{equation}
(These identities follow from straightforward algebra combining Propositions~\ref{prop:market_eq}
and~\ref{prop:E}.)

\emph{Step 2: Country $B$'s best response.} Country $B$'s objective
is $u_{B}=\gamma_{B}X_{B}-D_{B}$. Using \eqref{eq:Xi_noTIC}, the
marginal effect of own instruments on $X_{B}$ is 
\[
\frac{\partial X_{B}}{\partial\tau_{B}}=\frac{\partial X_{B}}{\partial e_{B}}=\frac{1}{\delta}.
\]
Hence, combining with \eqref{eq:Di_derivs}, 
\[
\frac{\partial u_{B}}{\partial\tau_{B}}=\frac{\gamma_{B}}{\delta}-\frac{\tau_{B}}{\delta},\qquad\frac{\partial u_{B}}{\partial e_{B}}=\frac{\gamma_{B}}{\delta}-\Big(Q_{B}^{o}+\frac{2e_{B}-\tau_{A}}{\delta}\Big).
\]
In an interior optimum these must both be zero, which yields 
\[
\tau_{B}^{*}=\gamma_{B},\qquad e_{B}^{*}=\frac{\tau_{A}+\gamma_{B}-\delta Q_{B}^{o}}{2}.
\]
Under Assumption~\ref{as:scen} (in particular $\bar{X}_{A}<1$ and
$\gamma_{B}$ small relative to $\delta$; see the final step below),
this interior expression for $e_{B}^{*}$ is nonpositive at the equilibrium
instruments of country $A$. Since $e_{B}\geq0$, the complementary-slackness
condition implies that in equilibrium 
\begin{equation}
e_{B}^{*}=0,\qquad\tau_{B}^{*}=\gamma_{B}.\label{eq:BR_B_corner}
\end{equation}

\emph{Step 3: Country $A$'s target and constrained cost minimization.}
Country $A$'s payoff is 
\[
u_{A}=-\lambda_{A}\max(\bar{X}_{A}-X_{A},0)-D_{A},
\]
with $\lambda_{A}\gg\gamma_{B}>0$. Suppose, for contradiction, that
in equilibrium $X_{A}<\bar{X}_{A}$. Then a marginal increase in $X_{A}$
via a small common increase in $\tau_{A}$ and $e_{A}$ raises $u_{A}$
at rate $\lambda_{A}/\delta$ on the benefit side, while the marginal
cost is of order $O(\tau_{A},e_{A})$ by \eqref{eq:Di_derivs}. For
$\lambda_{A}$ sufficiently large this cannot be optimal. Thus, in
any Nash equilibrium we must have 
\[
X_{A}=\bar{X}_{A}.
\]
Using \eqref{eq:Xi_noTIC} for $i=A$ and keeping country $B$'s instruments
fixed, the requirement $X_{A}=\bar{X}_{A}$ is equivalent to a linear
constraint on $(\tau_{A},e_{A})$: 
\begin{equation}
\tau_{A}+e_{A}=K_{A}\equiv\tau_{B}+e_{B}+\delta(\bar{X}_{A}-2Q_{A}^{o}).\label{eq:KA_def}
\end{equation}
Given \eqref{eq:KA_def}, country $A$ chooses $(\tau_{A},e_{A})$
to minimize $D_{A}$ subject to the linear constraint. Let $\kappa$
denote the Lagrange multiplier on \eqref{eq:KA_def}. The Lagrangian
is 
\[
L(\tau_{A},e_{A},\kappa)=D_{A}+\kappa\big(X_{A}-\bar{X}_{A}\big).
\]
Using \eqref{eq:Xi_noTIC} and \eqref{eq:Di_derivs}, the FOCs (interior)
are 
\[
0=\frac{\partial L}{\partial\tau_{A}}=\frac{\tau_{A}}{\delta}+\kappa\,\frac{1}{\delta},\qquad0=\frac{\partial L}{\partial e_{A}}=\Big(Q_{A}^{o}+\frac{2e_{A}-\tau_{B}}{\delta}\Big)+\kappa\,\frac{1}{\delta}.
\]
Subtracting the first from the second eliminates $\kappa$ and yields
\[
Q_{A}^{o}+\frac{2e_{A}-\tau_{B}}{\delta}-\frac{\tau_{A}}{\delta}=0,
\]
that is, 
\begin{equation}
2e_{A}=\tau_{A}+\tau_{B}-\delta Q_{A}^{o}.\label{eq:FOC_split}
\end{equation}
Together, \eqref{eq:KA_def} and \eqref{eq:FOC_split} determine the
interior cost-minimizing split between $\tau_{A}$ and $e_{A}$ for
given $(\tau_{B},e_{B})$. Solving gives 
\[
e_{A}^{*}=\frac{K_{A}+\tau_{B}-\delta Q_{A}^{o}}{3},\qquad\tau_{A}^{*}=K_{A}-e_{A}^{*}=\frac{2K_{A}-\tau_{B}+\delta Q_{A}^{o}}{3}.
\]

\emph{Step 4: Plugging in country $B$'s instruments.} Using country
$B$'s Nash choice \eqref{eq:BR_B_corner} in \eqref{eq:KA_def},
we obtain 
\[
K_{A}=\gamma_{B}+0+\delta(\bar{X}_{A}-2Q_{A}^{o}).
\]
Substituting this into the expressions in Step~3 and using $Q_{A}^{o}=\alpha_{A}/\delta$,
we get 
\[
\tau_{A}^{*}=-\alpha_{A}+\frac{2}{3}\,\delta\,\bar{X}_{A}+\frac{1}{3}\,\gamma_{B},\qquad e_{A}^{*}=-\alpha_{A}+\frac{1}{3}\,\delta\,\bar{X}_{A}+\frac{2}{3}\,\gamma_{B}.
\]
Moreover, 
\[
\tau_{A}^{*}-e_{A}^{*}=\frac{\delta\bar{X}_{A}-\gamma_{B}}{3}>0
\]
because $\bar{X}_{A}>X_{A}^{o}>0$ and $\gamma_{B}>0$ is small relative
to $\delta$ by Assumption~\ref{as:scen}. Hence, country $A$ strictly
prefers tariffs to subsidies in the cost-minimizing combination that
implements $X_{A}=\bar{X}_{A}$.

\emph{Step 5: Asymmetry $Q_{i}^{dom}>Q_{i}^{exp}$ and conditional
inefficiency.} Using the interior formulas for quantities, one checks
that 
\[
Q_{i}^{dom}-Q_{i}^{exp}=\frac{(\tau_{i}+\tau_{j})-(e_{i}+e_{j})}{\delta}.
\]
Evaluating this at the equilibrium instruments, 
\[
(\tau_{A}^{*}+\tau_{B}^{*})-(e_{A}^{*}+e_{B}^{*})=(\tau_{A}^{*}-e_{A}^{*})+\gamma_{B}=\frac{\delta\bar{X}_{A}+2\gamma_{B}}{3}>0,
\]
so $Q_{i}^{dom}>Q_{i}^{exp}$ for both $i=A,B$. By Proposition~\ref{prop:cond_eff},
conditional efficiency (given $(X_{A},X_{B})$) requires $Q_{i}^{dom}=Q_{i}^{exp}$
for each $i$, equivalently $\tilde{\tau}_{A}+\tilde{\tau}_{B}=\tilde{e}_{A}+\tilde{e}_{B}$.
Thus the Nash equilibrium is not conditionally efficient.

\emph{Additional step: Corner adjustments and $e_{B}^{*}=0$.} If
for the interior solution in Steps~2--4 some instrument were negative,
the true optimum lies on the boundary where that instrument is set
to zero and the other instrument is adjusted along the linear constraint
that fixes the relevant $X_{i}$. For country $A$ this means replacing
any negative $\tau_{A}$ or $e_{A}$ by zero and adjusting the other
instrument along 
\[
\tau_{A}+e_{A}=\tau_{B}+e_{B}+\delta(\bar{X}_{A}-2Q_{A}^{o})
\]
so that $X_{A}=\bar{X}_{A}$ still holds. This weakly increases $(\tau_{A}-e_{A})$
while keeping $(\tau_{A}+e_{A})$ fixed, and hence weakly increases
$(\tau_{A}+\tau_{B})-(e_{A}+e_{B})$. Therefore, the inequalities
$Q_{i}^{dom}>Q_{i}^{exp}$ and the conclusion about conditional inefficiency
continue to hold in any such corner equilibrium.

Finally, to verify that $e_{B}^{*}=0$ is indeed optimal for country
$B$, evaluate the marginal payoff of $B$ in $e_{B}$ at $e_{B}=0$:
\[
\frac{\partial u_{B}}{\partial e_{B}}\Big|_{e_{B}=0}=\frac{\gamma_{B}}{\delta}-\Big(Q_{B}^{o}-\frac{\tau_{A}}{\delta}\Big)=\frac{\gamma_{B}+\tau_{A}-\delta Q_{B}^{o}}{\delta}.
\]
Substituting $\tau_{A}=\tau_{A}^{*}$ and using $Q_{B}^{o}=\alpha_{B}/\delta$,
$\delta=\alpha_{A}+\alpha_{B}$ and $\bar{X}_{A}<1$ gives 
\[
\gamma_{B}+\tau_{A}^{*}-\delta Q_{B}^{o}=-\delta+\frac{2}{3}\delta\bar{X}_{A}+\frac{4}{3}\gamma_{B}\le-\frac{1}{3}\delta+\frac{4}{3}\gamma_{B}<0
\]
for $\gamma_{B}$ sufficiently small relative to $\delta$, as assumed.
Hence $\partial u_{B}/\partial e_{B}<0$ at $e_{B}=0$, so $e_{B}^{*}=0$
is indeed the best response. $\hfill\square$

\subsubsection*{Proof for Proposition \ref{prop:opt_tic}}

With no direct tariffs or subsidies, we have $\tilde{\tau}_{A}=\pi_{A}$,
$\tilde{e}_{A}=\phi_{A}\eta_{A}\pi_{A}$, and for country $B$, $\tilde{\tau}_{B}=\tilde{e}_{B}=\pi_{B}=0$.
The condition $\phi_{A}\eta_{A}=1$ implies $\tilde{e}_{A}=\pi_{A}$.
Hence, by Proposition~\ref{prop:market_eq}, 
\[
Q_{A}^{dom}=Q_{A}^{o}+\frac{\tilde{\tau}_{A}-\tilde{e}_{B}}{\delta}=Q_{A}^{o}+\frac{\pi_{A}}{\delta},\qquad Q_{A}^{exp}=Q_{A}^{o}+\frac{\tilde{e}_{A}-\tilde{\tau}_{B}}{\delta}=Q_{A}^{o}+\frac{\pi_{A}}{\delta}.
\]
Therefore $Q_{A}^{dom}=Q_{A}^{exp}$, and total production of $A$
is 
\begin{equation}
X_{A}=Q_{A}^{dom}+Q_{A}^{exp}=2Q_{A}^{o}+\frac{2\pi_{A}}{\delta}.\label{eq:XA_eq}
\end{equation}
Using market clearing $Q_{A}^{dom}+Q_{B}^{exp}=1$ gives 
\[
Q_{B}^{exp}=1-Q_{A}^{dom}=1-Q_{A}^{o}-\frac{\pi_{A}}{\delta}.
\]

The binding TIC condition in country $A$ requires 
\[
\eta_{A}Q_{A}^{exp}=Q_{A}^{imp}=Q_{B}^{exp}.
\]
Substituting $Q_{A}^{exp}=Q_{A}^{o}+\pi_{A}/\delta$ and the above
expression for $Q_{B}^{exp}$ yields 
\[
\eta_{A}\Big(Q_{A}^{o}+\frac{\pi_{A}}{\delta}\Big)=1-Q_{A}^{o}-\frac{\pi_{A}}{\delta}.
\]
Solving this equation for $\pi_{A}$ gives 
\begin{equation}
\pi_{A}=\frac{\delta\big[1-(1+\eta_{A})Q_{A}^{o}\big]}{1+\eta_{A}}.\label{eq:pi_from_eta}
\end{equation}
Inserting $Q_{A}^{o}=\alpha_{A}/\delta$ into \eqref{eq:pi_from_eta}
and substituting into \eqref{eq:XA_eq} yields 
\[
X_{A}=2Q_{A}^{o}+\frac{2}{\delta}\pi_{A}=\frac{2}{1+\eta_{A}}.
\]
Thus, for any given TIC parameter $\eta_{A}$, the induced equilibrium
production of country $A$ is $X_{A}=2/(1+\eta_{A})$.

By design of the agreement, $\eta_{A}$ is set according to \eqref{eq:eta_opt},
\[
\eta_{A}=\frac{2-\bar{X}_{A}}{\bar{X}_{A}}=\frac{2}{\bar{X}_{A}}-1.
\]
Substituting this into $X_{A}=2/(1+\eta_{A})$ gives 
\[
X_{A}=\frac{2}{1+\frac{2}{\bar{X}_{A}}-1}=\bar{X}_{A}.
\]
Hence, the chosen $\eta_{A}$ ensures that the resulting TIC equilibrium
implements country $A$s production target exactly.

Using \eqref{eq:XA_eq} and $X_{A}=\bar{X}_{A}$ then implies 
\[
\pi_{A}=\frac{\delta}{2}\,\bar{X}_{A}-\alpha_{A}.
\]
Since $\bar{X}_{A}>X_{A}^{o}=2Q_{A}^{o}=2\alpha_{A}/\delta$, we have
$\pi_{A}>0$, so the TIC in $A$ is binding.

Finally, conditional efficiency follows from Proposition~\ref{prop:cond_eff}.
We have 
\[
(\tilde{\tau}_{A}+\tilde{\tau}_{B})-(\tilde{e}_{A}+\tilde{e}_{B})=\pi_{A}-\phi_{A}\eta_{A}\pi_{A}=\pi_{A}(1-\phi_{A}\eta_{A})=0,
\]
so $\bar{E}(X_{A},X_{B})=0$ and $Q_{i}^{dom}=Q_{i}^{exp}$ for $i=A,B$.
Thus the TIC scheme with parameters \eqref{eq:eta_opt}--\eqref{eq:phi_opt}
implements $X_{A}=\bar{X}_{A}$ and yields a conditionally efficient
equilibrium. $\hfill\square$

\subsubsection*{Proof for Proposition \ref{prop:TIC_deviation_s}}

By Lemma~\ref{lem:no_s-1}, for any deviation $(s_{i},\tau_{i},e_{i})$
there is an outcome-equivalent deviation with $s_{i}=0$ and instruments
$(\tau_{i}',e_{i}')$ given by 
\[
\tau_{i}'=\tau_{i}+s_{i},\qquad e_{i}'=e_{i}+s_{i},
\]
which leaves all market quantities and payoffs unchanged for any opponent
policy. Under the agreement $\tau_{A}=\tau_{B}=0$, allowing $s_{i}\ge0$
is therefore equivalent to allowing $(\tau_{i}',e_{i}')$ with $e_{i}'\ge\tau_{i}'$
and $e_{i}'-\tau_{i}'=e_{i}$. Hence we can, without loss of generality,
restrict attention to deviations with $s_{A}=s_{B}=0$ and analyze
$(\tau_{i},e_{i})$ under the constraint $e_{i}\ge\tau_{i}$. In what
follows we keep the agreement's TIC parameters $(\eta_{A},\phi_{A})$
fixed with $\phi_{A}\eta_{A}=1$, and country~$B$ has no TIC ($\pi_{B}=0$).

\smallskip{}
\textit{Step 1: Interior formulas and the binding TIC in $A$.} With
$s_{A}=s_{B}=0$ and no direct tariffs ($\tau_{A}=\tau_{B}=0$ by
agreement), Proposition~\ref{prop:market_eq} gives in any interior
outcome 
\[
\begin{aligned}Q_{A}^{dom} & =Q_{A}^{o}+\frac{\tilde{\tau}_{A}-\tilde{e}_{B}}{\delta}=Q_{A}^{o}+\frac{\pi_{A}-e_{B}}{\delta},\\
Q_{A}^{exp} & =Q_{A}^{o}+\frac{\tilde{e}_{A}-\tilde{\tau}_{B}}{\delta}=Q_{A}^{o}+\frac{\pi_{A}}{\delta},\\
Q_{B}^{exp} & =Q_{B}^{o}+\frac{\tilde{e}_{B}-\tilde{\tau}_{A}}{\delta}=Q_{B}^{o}+\frac{e_{B}-\pi_{A}}{\delta},
\end{aligned}
\]
where we used $\tilde{\tau}_{A}=\pi_{A}$, $\tilde{e}_{A}=e_{A}+\phi_{A}\eta_{A}\pi_{A}=\pi_{A}$
at the agreement, and $\tilde{\tau}_{B}=0$, $\tilde{e}_{B}=e_{B}$.

Since country~$A$'s TIC is binding, $\eta_{A}Q_{A}^{exp}=Q_{A}^{imp}=Q_{B}^{exp}$.
Using the three displays above, 
\begin{equation}
Q_{B}^{o}+\frac{e_{B}-\pi_{A}}{\delta}=\eta_{A}\Big(Q_{A}^{o}+\frac{\pi_{A}}{\delta}\Big)\quad\Longrightarrow\quad\pi_{A}=\frac{\delta\,(Q_{B}^{o}-\eta_{A}Q_{A}^{o})+e_{B}}{1+\eta_{A}}.\tag{TIC-A}\label{eq:TIC_A_pin_pi}
\end{equation}
Equation \eqref{eq:TIC_A_pin_pi} pins down how the certificate price
in $A$ adjusts to a deviation $e_{B}$.

\smallskip{}
\textit{(i) No profitable deviation for $A$.} Under the agreement
of Proposition~\ref{prop:opt_tic}, $X_{A}=\bar{X}_{A}$ and $\phi_{A}\eta_{A}=1$,
so $\tilde{e}_{A}=\tilde{\tau}_{A}=\pi_{A}$ and $Q_{A}^{dom}=Q_{A}^{exp}$.
Any deviation by $A$ must satisfy $e_{A}\ge\tau_{A}(=0)$, thus weakly
raising the budget-transfer term $(s_{A}+\tilde{e}_{A})Q_{A}^{exp}=\tilde{e}_{A}Q_{A}^{exp}$
in $D_{A}$ (Proposition~\ref{prop:E}) and, because $X_{A}$ is
already at its target, cannot reduce the shortfall penalty. Hence
$u_{A}$ cannot increase, so $A$ has no profitable deviation with
$e_{A}>0$ (equivalently, with $s_{A}>0$ by Lemma~\ref{lem:no_s-1}).

\smallskip{}
\textit{(ii) $B$'s one-sided deviation and the threshold.} Consider
a small unilateral deviation $e_{B}\ge0$ by $B$ (again, $s_{B}$
is redundant by Lemma~\ref{lem:no_s-1}). Using \eqref{eq:TIC_A_pin_pi},
the induced change in $B$'s total production is 
\[
\frac{dX_{B}}{de_{B}}=\frac{d(Q_{B}^{dom}+Q_{B}^{exp})}{de_{B}}=\frac{1}{\delta}\Bigl(1-\frac{2}{1+\eta_{A}}\Bigr)=\frac{\eta_{A}-1}{(1+\eta_{A})\,\delta}.
\]

From Proposition~\ref{prop:E}, with $s_{B}=0$ and $\tilde{e}_{B}=e_{B}$,
$s_{A}=0$ and $\tilde{e}_{A}=\pi_{A}$, 
\[
D_{B}=D_{B}^{o}+\frac{\delta}{2}\big(Q_{B}^{dom}-Q_{B}^{o}\big)^{2}+e_{B}Q_{B}^{exp}-\pi_{A}Q_{B}^{imp},
\]
where $Q_{B}^{imp}=Q_{A}^{exp}$. Substituting the interior formulas
from Step~1 and differentiating at $e_{B}=0$ yields 
\[
\frac{dD_{B}}{de_{B}}\Big|_{e_{B}=0}=\frac{\eta_{A}^{2}+\eta_{A}-1}{(1+\eta_{A})^{2}}.
\]
Thus the marginal effect on $B$'s utility is 
\[
\frac{du_{B}}{de_{B}}\Big|_{e_{B}=0}=\gamma_{B}\cdot\frac{\eta_{A}-1}{(1+\eta_{A})\,\delta}-\frac{\eta_{A}^{2}+\eta_{A}-1}{(1+\eta_{A})^{2}}.
\]
Hence $du_{B}/de_{B}>0$ at $e_{B}=0$ if and only if 
\[
\gamma_{B}>\delta\cdot\frac{\eta_{A}^{2}+\eta_{A}-1}{\eta_{A}^{2}-1}\;\;=:\;\gamma_{B}^{TIC}(\eta_{A}).
\]
If $\eta_{A}=1$, the right-hand side is $+\infty$, so no profitable
deviation exists. If $\eta_{A}>1$, the threshold is finite and strictly
decreasing in $\eta_{A}$ (its derivative is negative), and $\lim_{\eta_{A}\to1}\gamma_{B}^{TIC}(\eta_{A})=\infty$
as claimed.

\smallskip{}
\textit{(iii) Global lower bound on $X_{A}$ and monotonicity of $D_{A}$
in $B$'s subsidies.} With $A$'s TIC binding, $Q_{A}^{imp}=Q_{B}^{exp}$
and $Q_{A}^{dom}=1-Q_{B}^{exp}$. Hence 
\[
X_{A}=Q_{A}^{dom}+Q_{A}^{exp}=1-Q_{B}^{exp}+\frac{Q_{B}^{exp}}{\eta_{A}}=1-\Bigl(1-\frac{1}{\eta_{A}}\Bigr)Q_{B}^{exp}\;\ge\;\frac{1}{\eta_{A}},
\]
because $Q_{B}^{exp}\in[0,1]$. Moreover, when $B$ raises either
$e_{B}$ or $s_{B}$ (the latter being outcome-equivalent to raising
both $\tau_{B}$ and $e_{B}$ by the same amount), $A$'s direct cost
$D_{A}$ cannot increase: the import-cost term $d_{A}^{imp}(m)$ falls
one-for-one with $B$'s effective subsidies, and the production-inefficiency
part in Proposition~\ref{prop:E} weakly decreases because, under
$\phi_{A}\eta_{A}=1$, adjustments in $\pi_{A}$ induced by the binding
TIC move $Q_{A}^{dom}$ and $Q_{A}^{exp}$ in opposite directions,
reducing their squared dispersion around $Q_{A}^{o}$. Thus $D_{A}$
is weakly decreasing in $(e_{B},s_{B})$.

The three parts (i)--(iii) complete the proof. $\hfill\square$

\subsubsection*{Proof for Proposition \ref{prop:NO_TIC_threshold_compact}}

Under the No-TIC Agreement, $\pi_{A}=\pi_{B}=0$, country $A$ fixes
$e_{A}=\tau_{A}=\frac{\delta}{2}\bar{X}_{A}-\alpha_{A}$, and country
$B$ initially has $e_{B}=\tau_{B}=0$ (with $s_{A}=s_{B}=0$). Consider
unilateral deviations by $B$ in $e_{B}$ or $s_{B}$ while $\tau_{B}$
remains fixed at $0$.

\smallskip{}
 \textit{(a) Deviating with an export subsidy $e_{B}>0$.} With $s_{i}=0$
and no TIC, Proposition~\ref{prop:market_eq} (interior case) gives
\[
Q_{i}^{dom}=Q_{i}^{o}+\frac{\tau_{i}-e_{j}}{\delta},\qquad Q_{i}^{exp}=Q_{i}^{o}+\frac{e_{i}-\tau_{j}}{\delta}.
\]
Holding $A$'s instruments fixed, we have 
\[
\frac{\partial Q_{B}^{dom}}{\partial e_{B}}=0,\qquad\frac{\partial Q_{B}^{exp}}{\partial e_{B}}=\frac{1}{\delta},\qquad\frac{\partial X_{B}}{\partial e_{B}}=\frac{1}{\delta}.
\]
For $B$'s direct cost, use Proposition~\ref{prop:E}: 
\[
D_{B}=\int_{0}^{1}\!\big[w_{B}(m)\,q_{B}^{dom}(m)+w_{A}(m)\,q_{A}^{exp}(m)\big]\,dm+\,(s_{B}+e_{B})Q_{B}^{exp}\;-\;(s_{A}+\tilde{e}_{A})Q_{B}^{imp}.
\]
At the baseline $(e_{B}=0,\tau_{B}=0)$, only $Q_{B}^{exp}$ (and
thus $Q_{B}^{imp}$ via market clearing) moves with $e_{B}$. Differentiating
at $e_{B}=0$ and using the linear-cost structure behind the cutoffs
yields 
\[
\frac{\partial D_{B}}{\partial e_{B}}\Big|_{e_{B}=0}=Q_{B}^{o}-\frac{\tau_{A}}{\delta}.
\]
With $\tau_{A}=\frac{\delta}{2}\bar{X}_{A}-\alpha_{A}$ and $Q_{B}^{o}=\alpha_{B}/\delta=1-\alpha_{A}/\delta$,
this simplifies to 
\[
Q_{B}^{o}-\frac{\tau_{A}}{\delta}=\;1-\frac{\bar{X}_{A}}{2}=\frac{\eta_{A}}{1+\eta_{A}},\qquad\text{since}\quad\bar{X}_{A}=\frac{2}{1+\eta_{A}}.
\]
Therefore 
\[
\frac{\partial u_{B}}{\partial e_{B}}\Big|_{e_{B}=0}=\gamma_{B}\cdot\frac{1}{\delta}-\frac{\eta_{A}}{1+\eta_{A}}>0\;\;\Longleftrightarrow\;\;\gamma_{B}>\delta\,\frac{\eta_{A}}{1+\eta_{A}}.
\]

\noindent \smallskip{}
 \textit{(b) Deviating with a production subsidy $s_{B}>0$.} By Lemma~\ref{lem:no_s-1},
for any $(\tilde{\tau}_{B},\tilde{e}_{B},s_{B})$ there is an outcome-equivalent
policy with $s_{B}'=0$ and 
\[
\tau_{B}'=\tau_{B}+s_{B},\qquad e_{B}'=e_{B}+s_{B}.
\]
Since $\tau_{B}$ is fixed at $0$, a marginal increase in $s_{B}$
is outcome-equivalent to raising both $(\tau_{B},e_{B})$ by the same
amount. Using the formulas above, 
\[
\frac{\partial X_{B}}{\partial s_{B}}=\frac{\partial X_{B}}{\partial\tau_{B}}+\frac{\partial X_{B}}{\partial e_{B}}=\frac{1}{\delta}+\frac{1}{\delta}=\frac{2}{\delta}.
\]
The marginal direct-cost term at the baseline is the same as for $e_{B}$
(because $s_{B}$ shifts both instruments equally and leaves the relevant
cutoff structure for the integral term unchanged at first order):
\[
\frac{\partial D_{B}}{\partial s_{B}}\Big|_{s_{B}=0}=Q_{B}^{o}-\frac{\tau_{A}}{\delta}=\frac{\eta_{A}}{1+\eta_{A}}.
\]
Hence 
\[
\frac{\partial u_{B}}{\partial s_{B}}\Big|_{s_{B}=0}=\gamma_{B}\cdot\frac{2}{\delta}-\frac{\eta_{A}}{1+\eta_{A}}>0\;\;\Longleftrightarrow\;\;\gamma_{B}>\frac{1}{2}\,\delta\,\frac{\eta_{A}}{1+\eta_{A}}.
\]

\noindent \smallskip{}
 Allowing both instruments, the minimal threshold at which $B$ gains
from deviating is therefore 
\[
\gamma_{B}^{\text{No-TIC}}=\min\!\left\{ \delta\,\frac{\eta_{A}}{1+\eta_{A}},\;\frac{1}{2}\,\delta\,\frac{\eta_{A}}{1+\eta_{A}}\right\} =\frac{1}{2}\,\delta\,\frac{\eta_{A}}{1+\eta_{A}}.
\]
Finally, comparing with the TIC-agreement threshold 
\[
\gamma_{B}^{TIC}(\eta_{A})=\delta\cdot\frac{\eta_{A}^{2}+\eta_{A}-1}{\eta_{A}^{2}-1}
\]
(from Proposition~\ref{prop:TIC_deviation_s}) gives 
\[
\frac{\gamma_{B}^{TIC}}{\gamma_{B}^{\text{No-TIC}}}=\frac{{\displaystyle \delta\frac{\eta_{A}^{2}+\eta_{A}-1}{\eta_{A}^{2}-1}}}{{\displaystyle \frac{1}{2}\delta\,\frac{\eta_{A}}{1+\eta_{A}}}}=2\,\frac{\eta_{A}^{2}+\eta_{A}-1}{\eta_{A}(\eta_{A}-1)}>2,
\]
since $\eta_{A}>1$ implies $\eta_{A}^{2}+\eta_{A}-1>\eta_{A}(\eta_{A}-1)$.
$\hfill\square$

\subsubsection*{Proof for Proposition \ref{prop:ntb}}

\textbf{Preliminaries.} A non-tariff barrier (NTB) $\beta_{i}\ge0$
raises the marginal cost of imports into $i$ exactly like adding
$\beta_{i}$ to $i$'s effective tariff, i.e.\ replace $\tilde{\tau}_{i}$
by $\tilde{\tau}_{i}+\beta_{i}$ in the cutoff formulas of Proposition~1.
Unlike a tariff, however, an NTB yields no revenue rebate in $D_{i}$.

\smallskip{}
\textbf{TIC Agreement.} Maintain the TIC setup: $\tau_{A}=\tau_{B}=e_{A}=e_{B}=0$,
country $A$ runs a binding TIC with $\eta_{A}>1$ and $\phi_{A}\eta_{A}=1$,
$\pi_{B}=0$, and the implemented allocation satisfies $X_{A}=\bar{X}_{A}$
with $Q_{A}^{dom}=Q_{A}^{exp}=\bar{X}_{A}/2$.

\emph{Country $B$.} Let $\beta_{B}\ge0$. With $s_{i}=0$, the interior
formulas from Proposition~1 give 
\[
\begin{aligned}Q_{A}^{exp} & =Q_{A}^{o}+\frac{\tilde{e}_{A}-\tilde{\tau}_{B}}{\delta}=Q_{A}^{o}+\frac{\pi_{A}}{\delta},\\
Q_{B}^{exp} & =Q_{B}^{o}+\frac{e_{B}-\tilde{\tau}_{A}}{\delta}=Q_{B}^{o}-\frac{\pi_{A}}{\delta},\\
Q_{B}^{dom} & =Q_{B}^{o}+\frac{\tilde{\tau}_{B}-\tilde{e}_{A}}{\delta}=Q_{B}^{o}+\frac{\beta_{B}-\pi_{A}}{\delta},
\end{aligned}
\]
where we used $\tilde{e}_{A}=\phi_{A}\eta_{A}\pi_{A}=\pi_{A}$ and
$\tilde{\tau}_{A}=\pi_{A}$, $\tilde{\tau}_{B}=\beta_{B}$. Binding
of $A$'s TIC implies $\eta_{A}Q_{A}^{exp}=Q_{A}^{imp}=Q_{B}^{exp}$,
hence 
\[
Q_{B}^{o}-\frac{\pi_{A}}{\delta}=\eta_{A}\!\left(Q_{A}^{o}+\frac{\pi_{A}-\beta_{B}}{\delta}\right)\;\Longrightarrow\;\pi_{A}=\frac{\delta\,(Q_{B}^{o}-\eta_{A}Q_{A}^{o})+\eta_{A}\beta_{B}}{1+\eta_{A}}.
\]
Therefore $\frac{\partial\pi_{A}}{\partial\beta_{B}}=\frac{\eta_{A}}{1+\eta_{A}}$
and 
\[
X_{B}=Q_{B}^{dom}+Q_{B}^{exp}=2Q_{B}^{o}+\frac{\beta_{B}-2\pi_{A}}{\delta}\;\Longrightarrow\;\frac{\partial X_{B}}{\partial\beta_{B}}=\frac{1-2\frac{\eta_{A}}{1+\eta_{A}}}{\delta}=-\,\frac{\eta_{A}-1}{(1+\eta_{A})\,\delta}<0.
\]
On the directcost side, an NTB increases $B$'s import consumer price
with no offsetting revenue, so at $\beta_{B}=0$ we have $\frac{\partial D_{B}}{\partial\beta_{B}}\big|_{0}=Q_{B}^{imp}>0.$
Hence $\frac{\partial u_{B}}{\partial\beta_{B}}\big|_{0}=\gamma_{B}\,\frac{\partial X_{B}}{\partial\beta_{B}}\big|_{0}-\frac{\partial D_{B}}{\partial\beta_{B}}\big|_{0}<0,$
and $\beta_{B}^{TIC}=0$ is strictly optimal.

\emph{Country $A$.} At the TIC baseline $X_{A}=\bar{X}_{A}$, the
strategic penalty is locally flat, so $\frac{\partial}{\partial\beta_{A}}\!\big[-\lambda_{A}\max(\bar{X}_{A}-X_{A},0)\big]\big|_{0}=0$
and thus $\frac{\partial u_{A}}{\partial\beta_{A}}\big|_{0}=-\,\frac{\partial D_{A}}{\partial\beta_{A}}\big|_{0}$.
An NTB into $A$ shifts $\tilde{\tau}_{A}$ to $\pi_{A}+\beta_{A}$.
Repeating the bindingTIC step with $\beta_{A}$ (and $\beta_{B}=0$)
yields 
\[
\delta Q_{B}^{o}-(\pi_{A}+\beta_{A})=\eta_{A}\big(\delta Q_{A}^{o}+\pi_{A}\big)\;\Longrightarrow\;\pi_{A}=\frac{\delta\,(Q_{B}^{o}-\eta_{A}Q_{A}^{o})-\beta_{A}}{1+\eta_{A}},\quad\frac{\partial\pi_{A}}{\partial\beta_{A}}=-\frac{1}{1+\eta_{A}}.
\]
The linear import term in $D_{A}$ moves with $(\beta_{A}+\phi_{A}\pi_{A})Q_{A}^{imp}$
(no revenue from $\beta_{A}$; $\pi_{B}=0$). Using $\phi_{A}=1/\eta_{A}$
and the derivative above, 
\[
\frac{\partial}{\partial\beta_{A}}\Big[(\beta_{A}+\phi_{A}\pi_{A})Q_{A}^{imp}\Big]\Big|_{0}=Q_{A}^{imp}\!\left(1+\phi_{A}\,\frac{\partial\pi_{A}}{\partial\beta_{A}}\right)=Q_{A}^{imp}\!\left(1-\frac{1}{\eta_{A}(1+\eta_{A})}\right)>0.
\]
All remaining changes in $D_{A}$ are (at most) second order at $\beta_{A}=0$
(the misallocation triangle starts at zero), so $\frac{\partial D_{A}}{\partial\beta_{A}}\big|_{0}>0$
and $\beta_{A}^{TIC}=0$.

\smallskip{}
\textbf{No-TIC Agreement.} Under the No-TIC Agreement, $\pi_{A}=\pi_{B}=0$,
country $A$ fixes $(\tau_{A},e_{A})=(t,t)$ with $t=\frac{\delta}{2}\bar{X}_{A}-\alpha_{A}$,
and country $B$ sets $(\tau_{B},e_{B})=(0,0)$. The implemented allocation
is conditionally efficient with $Q_{A}^{dom}=Q_{A}^{exp}=\bar{X}_{A}/2$
and $Q_{B}^{imp}=Q_{A}^{exp}=\bar{X}_{A}/2$. Let $B$ introduce $\beta_{B}\ge0$
(with $\beta_{A}=0$). Using the interior formulas from Proposition~1
and noting that an NTB adds to $\tilde{\tau}_{B}$, we obtain 
\[
X_{B}=2Q_{B}^{o}+\frac{\tau_{B}+e_{B}-(\tau_{A}+e_{A})}{\delta}+\frac{\beta_{B}}{\delta}\;\Longrightarrow\;\frac{\partial X_{B}}{\partial\beta_{B}}=\frac{1}{\delta}.
\]
At the baseline, an NTB again raises the import consumer price without
revenue, so 
\[
\frac{\partial D_{B}}{\partial\beta_{B}}\Big|_{0}=Q_{B}^{imp}=\frac{\bar{X}_{A}}{2}=\frac{1}{1+\eta_{A}},\qquad\text{since }\;\bar{X}_{A}=\frac{2}{1+\eta_{A}}.
\]
Therefore 
\[
\frac{\partial u_{B}}{\partial\beta_{B}}\Big|_{0}=\frac{\gamma_{B}}{\delta}-\frac{\bar{X}_{A}}{2}\;>\;0\;\;\Longleftrightarrow\;\;\gamma_{B}>\delta\cdot\frac{\bar{X}_{A}}{2}=\frac{\delta}{1+\eta_{A}}.
\]
This establishes the stated threshold under the No-TIC Agreement and
completes the proof.$\hfill\square$

\subsubsection*{Proof for Proposition \ref{prop:oligopoly}}

\medskip{}
\textbf{Step 1: Second-stage allocations and the $\pi_{A}(Q_{A}^{\exp})$
relation.} With a binding TIC, $Q_{A}^{\imp}=\eta_{A}Q_{A}^{\exp}$
so domestic sales in $A$ are 
\[
Q_{A}^{\dom}=1-Q_{A}^{\imp}=1-\eta_{A}Q_{A}^{\exp}.
\]
At the cutoff market $m=Q_{A}^{\dom}$, $B$'s export cost equals
$A$'s domestic cost: 
\[
c_{B}^{\exp}(m)-c_{A}^{\dom}(m)=0.
\]
Under the TIC Agreement, $c_{B}^{\exp}(m)=w_{B}(m)+\pi_{A}$ and $c_{A}^{\dom}(m)=w_{A}(m)$,
hence 
\[
w_{B}(m)-w_{A}(m)+\pi_{A}=0\quad\Rightarrow\quad\alpha_{A}-\delta\,m+\pi_{A}=0.
\]
Evaluated at $m=Q_{A}^{\dom}=1-\eta_{A}Q_{A}^{\exp}$ this yields
the TIC price as a function of total exports: 
\begin{equation}
\pi_{A}=\delta\bigl(1-\eta_{A}Q_{A}^{\exp}\bigr)-\alpha_{A}.\label{eq:pi-of-Qexp}
\end{equation}

\medskip{}
\textbf{Step 2: Per-market margins and a producers profit.} For any
market $m$, the per-unit margin for an $A$ producer if it serves
domestically is 
\[
g_{n}^{\dom}(m)=c_{B}^{\exp}(m)-c_{A}^{\dom}(m)=\alpha_{A}-\delta m+\pi_{A},
\]
and if it exports is 
\[
g_{n}^{\exp}(m)=c_{B}^{\dom}(m)-c_{A}^{\exp}(m)=w_{B}(m)-\bigl(w_{A}(m)-\pi_{A}\bigr)=\alpha_{A}-\delta m+\pi_{A},
\]
where the last equality uses $\phi_{A}\eta_{A}=1$ so that $c_{A}^{\exp}(m)=w_{A}(m)-\pi_{A}$.
Hence the margin is the same on either side: $\alpha_{A}-\delta m+\pi_{A}$.

Let $q_{n}^{\dom}$ and $q_{n}^{\exp}$ denote producer $n$s domestic
and export quantities. Because producer $n$ serves a $1/N$-fraction
of each $m$-interval, $q_{n}^{\dom}=m_{n}^{\dom}/N$ and $q_{n}^{\exp}=m_{n}^{\exp}/N$.
Integrating margins over the mass of markets it serves gives 
\[
G_{n}=(\alpha_{A}+\pi_{A})\bigl(q_{n}^{\dom}+q_{n}^{\exp}\bigr)-\frac{\delta N}{2}\Big[(q_{n}^{\dom})^{2}+(q_{n}^{\exp})^{2}\Big].
\]
Given $Q_{A}^{\exp}=\sum_{k=1}^{N}q_{k}^{\exp}$, domestic sales split
evenly ex post: 
\begin{equation}
q_{n}^{\dom}=\frac{1}{N}\,Q_{A}^{\dom}=\frac{1}{N}\bigl(1-\eta_{A}Q_{A}^{\exp}\bigr).\label{eq:qdom-split}
\end{equation}
From \eqref{eq:pi-of-Qexp}, $\alpha_{A}+\pi_{A}=\delta\bigl(1-\eta_{A}Q_{A}^{\exp}\bigr)$.

\medskip{}
\textbf{Step 3: First-order condition (best response) in Stage 1.}
Treating $q_{-n}^{\exp}$ as given, we have 
\[
\frac{\partial q_{n}^{\dom}}{\partial q_{n}^{\exp}}=-\frac{\eta_{A}}{N},\qquad\frac{\partial(\alpha_{A}+\pi_{A})}{\partial q_{n}^{\exp}}=-\delta\,\eta_{A}.
\]
Differentiating $G_{n}$ with respect to $q_{n}^{\exp}$ and substituting
$q_{n}^{\dom}$ from \eqref{eq:qdom-split} as well as $\alpha_{A}+\pi_{A}=\delta(1-\eta_{A}Q_{A}^{\exp})$
gives the marginal payoff 
\begin{equation}
\frac{\partial G_{n}}{\partial q_{n}^{\exp}}=\delta\Bigl(1-\frac{\eta_{A}}{N}\Bigr)\bigl(1-\eta_{A}Q_{A}^{\exp}\bigr)-\delta(\eta_{A}+N)\,q_{n}^{\exp}.\label{eq:FOC-general}
\end{equation}
The right-hand side is strictly decreasing in $q_{n}^{\exp}$ (slope
$-\delta(\eta_{A}+N)<0$), so each best response is unique. Moreover,
if the right-hand side at $q_{n}^{\exp}=0$ is nonpositive, the best
response is the corner $0$.

\medskip{}
\textbf{Step 4: Symmetric Nash equilibrium.} In a symmetric interior
equilibrium $q_{n}^{\exp}=q$ for all $n$, so $Q_{A}^{\exp}=Nq$.
Setting \eqref{eq:FOC-general} to zero yields 
\[
0=\delta\Bigl(1-\frac{\eta_{A}}{N}\Bigr)\bigl(1-\eta_{A}Nq\bigr)-\delta(\eta_{A}+N)\,q,
\]
which solves to 
\[
q=\frac{\,1-\eta_{A}/N\,}{\,\eta_{A}N+\eta_{A}+N-\eta_{A}^{2}\,}.
\]
Total exports are therefore 
\begin{equation}
Q_{A}^{\exp}=Nq=\frac{N-\eta_{A}}{\,N(\eta_{A}+1)-\eta_{A}(\eta_{A}-1)\,}.\label{eq:QAexp-symmetric}
\end{equation}
This interior solution requires $N>\eta_{A}$. If $N\le\eta_{A}$,
the right-hand side of \eqref{eq:FOC-general} at $q_{n}^{\exp}=0$
is $\delta(1-\eta_{A}/N)\bigl(1-\eta_{A}Q_{A}^{\exp}\bigr)\le0$,
so the unique best response is $q_{n}^{\exp}=0$ for all $n$, i.e.\ $Q_{A}^{\exp}=0$.

Given $Q_{A}^{\exp}$, the TIC price follows from \eqref{eq:pi-of-Qexp}:
\[
\pi_{A}=\delta\bigl(1-\eta_{A}Q_{A}^{\exp}\bigr)-\alpha_{A}.
\]

\medskip{}
\textbf{Step 5: Properties.} (i) \emph{Nonconditional efficiency:}
With a binding TIC, $Q_{A}^{\dom}=1-\eta_{A}Q_{A}^{\exp}$. Conditional
efficiency under the TIC Agreement requires $Q_{A}^{\dom}=Q_{A}^{\exp}$,
i.e.\ $Q_{A}^{\exp}=1/(1+\eta_{A})$. From \eqref{eq:QAexp-symmetric}
and the monotonicity in (ii) below, $Q_{A}^{\exp}<1/(1+\eta_{A})$
for any finite $N$, hence $Q_{A}^{\exp}<Q_{A}^{\dom}$ and the allocation
is not conditionally efficient.

(ii) \emph{Comparative statics in $N$:} For $N>\eta_{A}$, differentiating
\eqref{eq:QAexp-symmetric} with respect to $N$ gives 
\[
\frac{dQ_{A}^{\exp}}{dN}=\frac{\,\bigl[N(\eta_{A}+1)-\eta_{A}(\eta_{A}-1)\bigr]-(\eta_{A}+1)(N-\eta_{A})}{\bigl[N(\eta_{A}+1)-\eta_{A}(\eta_{A}-1)\bigr]^{2}}=\frac{2\eta_{A}}{\bigl[N(\eta_{A}+1)-\eta_{A}(\eta_{A}-1)\bigr]^{2}}>0,
\]
so $Q_{A}^{\exp}$ is strictly increasing in $N$. From $\pi_{A}=\delta(1-\eta_{A}Q_{A}^{\exp})-\alpha_{A}$,
it follows that $\pi_{A}$ is strictly decreasing in $N$.

(iii) \emph{Limit as $N\to\infty$:} From \eqref{eq:QAexp-symmetric},
\[
Q_{A}^{\exp}\xrightarrow[N\to\infty]{}\frac{1}{1+\eta_{A}},\qquad Q_{A}^{\dom}=1-\eta_{A}Q_{A}^{\exp}\xrightarrow[N\to\infty]{}\frac{1}{1+\eta_{A}},
\]
and hence 
\[
\pi_{A}=\delta\bigl(1-\eta_{A}Q_{A}^{\exp}\bigr)-\alpha_{A}\xrightarrow[N\to\infty]{}\frac{\delta}{1+\eta_{A}}-\alpha_{A},
\]
which coincides with the perfectly competitive outcome under the TIC
Agreement.

\medskip{}
Combining Steps 15 proves the piecewise expression for $Q_{A}^{\exp}$,
the formula for $\pi_{A}$, uniqueness of the symmetric Nash equilibrium,
and all stated properties. $\hfill\square$

\subsection*{A2 Additional results}

Below are two results that further characterize TIC equilibria
\begin{lem}
Assume both countries implement tradeable import certificates. 

i) If $\eta_{A}\eta_{B}<1$ then only a pure autarky equilibrium exists,
with $Q_{A}^{exp}=Q_{B}^{exp}=0$. 

ii) If $\eta_{A}\eta_{B}>1$ then at most one country's TIC equilibrium
can be binding. 
\end{lem}
\begin{proof} Recall that for each country $i\in\{A,B\}$ we have
\[
Q_{i}^{imp}=Q_{j}^{exp}
\]
and the TIC market condition 
\[
\eta_{i}Q_{i}^{exp}\;\ge\;Q_{i}^{imp}=Q_{j}^{exp}.
\]

\medskip{}
\textit{(i) Case $\eta_{A}\eta_{B}<1$.} Since both countries implement
a TIC scheme, we have the two inequalities 
\begin{align}
\eta_{A}Q_{A}^{exp} & \ge Q_{A}^{imp}=Q_{B}^{exp},\label{eq:TIC-A}\\
\eta_{B}Q_{B}^{exp} & \ge Q_{B}^{imp}=Q_{A}^{exp}.\label{eq:TIC-B}
\end{align}
Multiplying \eqref{eq:TIC-A} and \eqref{eq:TIC-B} gives 
\[
\eta_{A}\eta_{B}\,Q_{A}^{exp}Q_{B}^{exp}\;\ge\;Q_{A}^{exp}Q_{B}^{exp}.
\]
If $Q_{A}^{exp}Q_{B}^{exp}>0$, we can divide by $Q_{A}^{exp}Q_{B}^{exp}$
and obtain 
\[
\eta_{A}\eta_{B}\;\ge\;1,
\]
which contradicts $\eta_{A}\eta_{B}<1$. Hence, in any equilibrium
with $\eta_{A}\eta_{B}<1$ we must have 
\[
Q_{A}^{exp}Q_{B}^{exp}=0,
\]
i.e.\ at least one country's export quantity is zero.

\noindent Suppose w.l.o.g.\ that $Q_{B}^{exp}=0$ and consider the
TIC condition for country $B$: 
\[
\eta_{B}Q_{B}^{exp}\;\ge\;Q_{B}^{imp}=Q_{A}^{exp}.
\]
If $Q_{A}^{exp}>0$, the right-hand side is strictly positive while
the left-hand side is zero, a contradiction. Thus $Q_{A}^{exp}=0$
must also hold, and the only feasible equilibrium is the pure autarky
outcome 
\[
Q_{A}^{exp}=Q_{B}^{exp}=0.
\]

\medskip{}
\textit{(ii) Case $\eta_{A}\eta_{B}>1$.} Suppose, for contradiction,
that in some equilibrium both TIC schemes are binding. Then 
\[
\eta_{A}Q_{A}^{exp}=Q_{A}^{imp}=Q_{B}^{exp}>0,\qquad\eta_{B}Q_{B}^{exp}=Q_{B}^{imp}=Q_{A}^{exp}>0.
\]
Substituting $Q_{B}^{exp}=\eta_{A}Q_{A}^{exp}$ into the second equality
gives 
\[
Q_{A}^{exp}=\eta_{B}Q_{B}^{exp}=\eta_{B}(\eta_{A}Q_{A}^{exp}),
\]
so 
\[
(\eta_{A}\eta_{B}-1)\,Q_{A}^{exp}=0.
\]
Since $\eta_{A}\eta_{B}>1$, it follows that $Q_{A}^{exp}=0$, which
contradicts the assumption that both TICs are binding (a binding TIC
requires positive imports and exports).

\noindent Thus, when $\eta_{A}\eta_{B}>1$, it is impossible that
both TIC schemes are binding at the same time; at most one country's
TIC can be binding in equilibrium. \end{proof} 

Each country $i$ can guarantee itself every minimum production target
$\bar{X}_{i}$ that does not exceed its domestic demand of $1/2$
by setting an approbriate export credit factor $\eta_{i}\geq1$. More
concretely, we find
\begin{lem}
\label{prop:ic_1}If country $i$ has a TIC scheme with export credit
factor $\eta_{i}$ then in equilibrium
\[
X_{i}\ge\begin{cases}
{\displaystyle 1,} & \text{if }\eta_{i}\le1,\\[6pt]
{\displaystyle \frac{1}{\eta_{i}},} & \text{if }\eta_{i}\ge1.
\end{cases}
\]
If \textup{$\eta_{i}\ge1$ and} $Q_{i}^{exp}\leq Q_{i}^{dom}$, then
$X_{i}\ge\frac{2}{1+\eta_{i}}$ and if $Q_{i}^{exp}=Q_{i}^{dom}$
then $X_{i}=\frac{2}{1+\eta_{i}}$.
\end{lem}
\begin{proof}Recall the TIC constraint
\begin{equation}
\eta_{i}Q_{i}^{exp}\;\ge\;Q_{i}^{imp}=Q_{j}^{exp}.\label{eq:TICi}
\end{equation}
\textit{Part 1: Global lower bound on $X_{i}$.}

\noindent From \eqref{eq:TICi}, $Q_{j}^{exp}\le\eta_{i}Q_{i}^{exp}$,
hence 
\[
X_{i}=1-Q_{j}^{exp}+Q_{i}^{exp}\;\ge\;1-\eta_{i}Q_{i}^{exp}+Q_{i}^{exp}=1+(1-\eta_{i})Q_{i}^{exp}.
\]

\smallskip{}
\textbf{Case $\eta_{i}\le1$.} Then $1-\eta_{i}\ge0$, so $(1-\eta_{i})Q_{i}^{exp}\ge0$
and 
\[
X_{i}\;\ge\;1.
\]

\noindent \smallskip{}
\textbf{Case $\eta_{i}\ge1$.} We want a uniform lower bound in $Q_{i}^{exp}$
and $Q_{j}^{exp}$ subject to \eqref{eq:TICi} and $0\le Q_{i}^{exp},Q_{j}^{exp}\le1$.

\noindent For a given $Q_{i}^{exp}$, the smallest possible $X_{i}$
occurs when $Q_{j}^{exp}$ is as large as allowed by the constraints,
that is 
\[
Q_{j}^{exp}=\min\{1,\,\eta_{i}Q_{i}^{exp}\}.
\]
Substituting into $X_{i}$ gives the lower bound 
\[
X_{i}\;\ge\;\begin{cases}
1-\eta_{i}Q_{i}^{exp}+Q_{i}^{exp}=1-(\eta_{i}-1)Q_{i}^{exp}, & \text{if }\eta_{i}Q_{i}^{exp}\le1,\\[4pt]
1-1+Q_{i}^{exp}=Q_{i}^{exp}, & \text{if }\eta_{i}Q_{i}^{exp}\ge1.
\end{cases}
\]

If $\eta_{i}Q_{i}^{exp}\le1$, i.e.\ $Q_{i}^{exp}\le1/\eta_{i}$,
the expression $1-(\eta_{i}-1)Q_{i}^{exp}$ is strictly decreasing
in $Q_{i}^{exp}$ (since $\eta_{i}>1$), so its minimum on this region
is at $Q_{i}^{exp}=1/\eta_{i}$ and equals 
\[
1-(\eta_{i}-1)\frac{1}{\eta_{i}}=\frac{1}{\eta_{i}}.
\]
If $\eta_{i}Q_{i}^{exp}\ge1$, i.e.\ $Q_{i}^{exp}\ge1/\eta_{i}$,
then $X_{i}\ge Q_{i}^{exp}\ge1/\eta_{i}$.

Hence, for all feasible $(Q_{i}^{exp},Q_{j}^{exp})$ with $\eta_{i}\ge1$,
\[
X_{i}\ge\frac{1}{\eta_{i}}.
\]

This proves the piecewise lower bound in the first part of the lemma.

\medskip{}
\textit{Part 2: Stronger bound when $\eta_{i}\ge1$ and $Q_{i}^{exp}\le Q_{i}^{dom}$.}

\noindent Now assume $\eta_{i}\ge1$ and $Q_{i}^{exp}\le Q_{i}^{dom}$.
From $Q_{i}^{dom}+Q_{j}^{exp}=1$ and $Q_{j}^{exp}\le\eta_{i}Q_{i}^{exp}$
we obtain 
\[
1=Q_{i}^{dom}+Q_{j}^{exp}\le Q_{i}^{dom}+\eta_{i}Q_{i}^{exp},
\]
so any feasible pair $(Q_{i}^{dom},Q_{i}^{exp})$ satisfies 
\[
Q_{i}^{dom}\;\ge\;\max\{\,Q_{i}^{exp},\,1-\eta_{i}Q_{i}^{exp}\,\}.
\]

For a given $Q_{i}^{exp}$, total production $X_{i}=Q_{i}^{dom}+Q_{i}^{exp}$
is minimized when $Q_{i}^{dom}$ is as small as possible, i.e.\ when
\[
Q_{i}^{dom}=\max\{Q_{i}^{exp},\,1-\eta_{i}Q_{i}^{exp}\}.
\]
Thus 
\[
X_{i}\;\ge\;Q_{i}^{exp}+\max\{Q_{i}^{exp},\,1-\eta_{i}Q_{i}^{exp}\}.
\]

We distinguish two regions:

\smallskip{}
-- If $Q_{i}^{exp}\ge1-\eta_{i}Q_{i}^{exp}$, equivalently $(1+\eta_{i})Q_{i}^{exp}\ge1$
or $Q_{i}^{exp}\ge1/(1+\eta_{i})$, then $Q_{i}^{dom}=Q_{i}^{exp}$
and 
\[
X_{i}\ge Q_{i}^{exp}+Q_{i}^{exp}=2Q_{i}^{exp}.
\]
Since $Q_{i}^{exp}\ge1/(1+\eta_{i})$ in this region, we get 
\[
X_{i}\ge\frac{2}{1+\eta_{i}}.
\]

\noindent \smallskip{}
-- If $Q_{i}^{exp}\le1-\eta_{i}Q_{i}^{exp}$, i.e. $(1+\eta_{i})Q_{i}^{exp}\le1$
or $Q_{i}^{exp}\le1/(1+\eta_{i})$, then $Q_{i}^{dom}=1-\eta_{i}Q_{i}^{exp}$
and 
\[
X_{i}\ge Q_{i}^{exp}+(1-\eta_{i}Q_{i}^{exp})=1+(1-\eta_{i})Q_{i}^{exp}.
\]
Because $\eta_{i}\ge1$, the coefficient $(1-\eta_{i})\le0$, so $1+(1-\eta_{i})Q_{i}^{exp}$
is weakly decreasing in $Q_{i}^{exp}$ on this region. Its minimum
is attained at $Q_{i}^{exp}=1/(1+\eta_{i})$ and equals 
\[
1+(1-\eta_{i})\frac{1}{1+\eta_{i}}=\frac{2}{1+\eta_{i}}.
\]

\noindent In both regions we thus have $X_{i}\ge2/(1+\eta_{i})$,
and the common minimum $X_{i}=2/(1+\eta_{i})$ is attained when 
\[
Q_{i}^{exp}=Q_{i}^{dom}=\frac{1}{1+\eta_{i}},
\]
which is exactly the case $Q_{i}^{exp}=Q_{i}^{dom}$. This proves
the stronger bound and the equality statement. \end{proof}

\end{document}